\documentclass[12pt,a4paper]{amsart}

\usepackage{hyperref}
\usepackage{braket}
\usepackage{amssymb}
\usepackage{amsmath}
\usepackage{amsthm}
\usepackage{amscd}
\usepackage{graphicx}
\usepackage[small,nohug,heads=vee]{diagrams}
\diagramstyle[labelstyle=\scriptstyle]
\usepackage{geometry}
\newtheorem{theorem}{Theorem}[section]
\newtheorem{lemma}[theorem]{Lemma}
\newtheorem{proposition}[theorem]{Proposition}

\theoremstyle{definition}
\newtheorem{definition}[theorem]{Definition}
\newtheorem{example}[theorem]{Example}

\theoremstyle{remark}

\newcommand{\lift}[2]{%
\setlength{\unitlength}{1pt}
\begin{picture}(0,0)(0,0)
\put(0,{#1}){\makebox(0,0)[b]{${#2}$}}
\end{picture}
}
\newcommand{\lowerarrow}[1]{%
\setlength{\unitlength}{0.03\DiagramCellWidth}
\begin{picture}(0,0)(0,0)
\qbezier(-28,-4)(0,-18)(28,-4)
\put(0,-14){\makebox(0,0)[t]{$\scriptstyle {#1}$}}
\put(28.6,-3.7){\vector(2,1){0}}
\end{picture}
}
\newcommand{\upperarrow}[1]{%
\setlength{\unitlength}{0.03\DiagramCellWidth}
\begin{picture}(0,0)(0,0)
\qbezier(-28,11)(0,25)(28,11)
\put(0,21){\makebox(0,0)[b]{$\scriptstyle {#1}$}}
\put(28.6,10.7){\vector(2,-1){0}}
\end{picture}
}

\begin{document}

\title{Transition probability spaces in loop quantum gravity}

\author{Xiao-Kan Guo}

\address{Wuhan Institute of Physics and Mathematics, Chinese Academy of Sciences, Wuhan 430071, China, and University of  Chinese Academy of Sciences, Beijing 100049, China}
\email{kankuohsiao@whu.edu.cn}





\date{\today}



\begin{abstract}
We study the (generalized) transition probability spaces, in the sense of Mielnik and Cantoni, for spacetime quantum states in loop quantum gravity. First, we show that loop quantum gravity admits the structures of transition probability spaces. This is exemplified by first checking such structures in covariant quantum mechanics, and then identify the transition probability spaces  in spin foam models via a simplified discrete version of general boundary formulation. The transition probability space thus defined gives a simple way to reconstruct the discrete analog of the  Hilbert space of the canonical theory and the relevant quantum logical structures. Second, we show that the transition probability space and in particular the spin foam model are  2-categories. Then we discuss how to realize in spin foam models two proposals by Crane about the mathematical structures of quantum gravity, namely the quantum topos and causal sites. We conclude that transition probability spaces provide us with an  alternative framework to understand various foundational questions of loop quantum gravity.
\end{abstract}
\maketitle


\section{Introduction}
Right after the advent of quantum mechanics (QM), the logical structures of QM was studied by  Birkhoff and von Neumann \cite{BvN36}. This lattice-theory-based quantum logic not only provides us with  a deep foundation of the Hilbert space formalism of QM \cite{Piron64}, but also leads us to a more natural mathematical structure for QM (namely, the type II$_1$ factor von Neumann algebra \cite{Red98}). This already indicates the possibility of describing QM by theories different from or even beyond  the theory of Hilbert spaces. Since, for both phenomenological and operational reasons, the most relevant quantities in QM are transition probabilities rather than the logical structures, a possible alternative description of QM can be based on the space of (generalized) transition probabilities \cite{Miel68,Bel76,Can75,Can82}.  The advantage of using transition probability spaces is that  the state space of QM are characterized directly by transition probabilities between quantum states, which are undefinable in the original quantum logic (for some examples of which see Ref. \cite{Gud81}) unless additional structures (such as continuous geometry) are  introduced into the framework.
Note that such a change of description does not mean that we have abandoned the ontologically logical structures and chosen an epistemologically phenomenological approach. In fact, the logical and probabilistic aspects of QM can be included in a unified framework \cite{Gun67}, and at a deeper level, they can be simultaneously and coherently defined on a lattice \cite{HSP14}. Therefore, it remains fundamental but at the same time has an operational meaning  in real experiments.

In this paper we will  study the theoretical structures of loop quantum gravity (LQG) with the help of transition probability spaces. Methodologically, since LQG is {\it within} the quantum formalism, these structures of transition probability spces are sure to be correct in LQG, however, they are not easy to be identified. What we hope is that this alternative description of the quantum formalism will clarify the structures of LQG. On the other hand, if a theory of quantum gravity is in general  {\it beyond} the current formalism of quantum theory, we think that the transition probability space  will be a possible approach to that more general theory.

A direct  motivation of this work comes from the spin foam approach to LQG \cite{Perez13}, where the crucial quantity is the spin foam amplitude, the transition amplitude between spin network states. The second motivation comes from the observation that the rigorous kinematic Hilbert space of canonical loop quantum gravity constructed from the holonomy-flux algebra and the quantum spin dynamics thereon \cite{Thiebook}  allow  further  mathematical constructions of a quantum logic or  a transition amplitude space. It is therefore hoped that these two approaches can be related or unified in a transition probability (or amplitude) space.
Besides, it is well-known that observables are difficult to define, even in classical general relativity, and the situation is not getting better in quantum gravity. Now in the approach via transition probability spaces, the fundamental concepts are shifted from the observables to the states and the transition probabilities between them, which is obviously more suitable for analyzing quantum gravity. 

The first goal of this paper is to describe spin foam models by transition probabilty spaces. To this end, we will take the following three steps of descriptions:
\begin{enumerate}
\item{\it General covariant QM}. This is because the transition probability spaces in the literature are defined only for nonrelativistic QM. In relativistic gravitational systems, general covariance is required. We thus consider the general covariant QM proposed by Reisenberger\&Rovelli \cite{RR02} as an example, and construct the corresponding transition probability space. We also consider the timeless path integral \cite{Chiou13} derived from the general covariant QM where the transition amplitudes have been constructed. Although the general covariant QM cannot be easily generalized to the full quantum gravity, this serves as an example to identify the transition probability space.
\item{\it General boundary formulation}. The general boundary formulation of quantum theories \cite{oe03a} provides a bridge connecting QM and quantum gravity. For the purpose of describing quantum gravity, we work in a simplified general boundary formulation  with the manifold structures replaced by discrete data on them.
We consider the probabilities in this simplified formulation and show that they define a transition probability space.
\item{\it Spin foam models}. When constructed from a group field theory, the transition amplitudes in a spin foam model formally define a transition probability space.
As is shown in Ref. \cite{CF08}, the EPRL(/FK) spin foam model can be cast into the general boundary formulation, which also manifests a description via transition probability spaces.
\end{enumerate}
That all these theories admit the descriptions via tansition probability spaces is not surprising, but the supporting arguments are far from obvious. A direct consequence of this description is the ability to reconstruct the canonical formalism from the transition amplitudes, which not only corroborates but simplifies earlier constructions. Then the logic of quantum gravity can be readily constructed from the transition probability space, and in particular we relate the quantum logical ``states" to transition probabilities to highlight the gauge freedoms in the spinfoam path integrals.
These are done in Sec. \ref{sec3}.

Concerning the logical approaches to quantum gravity, there is a way to classify many different logical formulations of quantum gravity by considering possible modifications of the classical 2-valued truth value map
\[\phi:\mathfrak{B}\rightarrow\mathbb{Z}_2\]
where $\mathfrak{B}$ is the Boolean algebra for classical logic and $\phi$ is a homomorphism of algebras. This consideration is first taken by Sorkin in Ref. \cite{Sor07} where the homomorphism $\phi$ is relaxed and the resulting ``quantum logic" is the anhomomorphic logic. This anhomomorphic logic is suitable for the causal sets approach to quantum gravity, which in a sense can be related to LQG via the quantum causal histories of spin networks \cite{MS97}.
One can also choose to modify the 2-value feature and change $\mathbb{Z}_2$ into an elementary topos. This results in the topos quantum theory, which has been applied to quantum gravity (although not in LQG, cf. \cite{Flo09}). On the other hand, The modification of the Boolean algebra $\mathfrak{B}$ into a Hilbert lattice, which occupies the name ``quantum logics", has not been sucessfully applied in quantum gravity. Previous attempts \cite{Ant94} resulted in (intuitionistic) logical structures  more general than quantum logics. In Sec. \ref{sec3}, we construct a quantum logical structre $L$ for LQG from the transition probability space of spin foams, but by calculating the Kolmogorov-Sinai entropy on $L$  we find this  $L$ admits a Lindenbaum-Tarski type algebra conforming to the previous results \cite{Ant94}. Instead of staying in logics, we next turn to the categorical aspects of such structures.

The second purpose of this paper is therfore to search for possible categorical characterizations of transition probability spaces in LQG. We also give explict realizations in spin foam models of several proposals given by L. Crane. We first show that, following the initiative of Baez \cite{Baez98},  the transition probability spaces and hence the spin foam models themselves are 2-categories. In particular, similar to but different from those in Ref. \cite{DMY10}, the 2-cells of the 2-category of spin foams are  identified to be the branched covers over spin netwoks from which the spin foams emanate. Based on these 2-categories, we  show that the property transitions in a quantum logic can be described by state tranitions in the classical Boolean algebra obtained through  localizing onto the presheaves, which realizes the structure of a quantum topos. Furthermore, we explore the possibility that causality in spin foam models is encoded in the causal category or causal sites of spin foams. These are collected in Sec. \ref{sec4}. 

 In the next section, we will briefly review the definitions and basic properties of transition probability spaces, their generalizations, and transition amplitude spaces. A new view of these structures from the operational probabilistic quantum theory is also given. Before that, let us recollect here some basic concepts in LQG.
\subsection{Basic LQG}
The canonical LQG \cite{Thiebook} on a four-dimensional manifold $M=\Sigma\times\mathbb{R}$ has the structure of an SU(2) gauge theory with Ashtekar-Barbero variables  $(A^i_a, E^a_i)$, where $A^i_a$ is an $\mathfrak{su}(2)$-valued connection 1-form and $E^a_i$ is a densitized triad from which the 3D metirc of $\Sigma$ is recovered as $h_{ab}=e_a^ie_b^j\delta_{ij}$ with $E^a_i=\frac{1}{2}\epsilon^{abc}\epsilon_{ijk}e^j_be^k_c\text{sgn}(\det e)$. Here $a,b,c,...$ are spatial indices and $i,j,k,...$ are $\mathfrak{su}(2)$ indices. The nontrivial Poisson brackets are
\[\{A^i_a(x),E^b_j(y)\}=\kappa\beta\delta^b_a\delta^i_j\delta^3(x,y),\quad\kappa=8\pi G,\]
where $\beta$ is the Barbero-Immirzi parameter. Then the Gaussian constraint $G_i=D_aE^a_i$ ensures the SU(2) gauge invariance, and the other constraints mean the invariance under spatial, i.e. $V_a=F^i_{ab}E^b_i$, and temporal, i.e. the Hamiltonian constraint $H$, diffeomorphisms.

The kinematical states of LQG lives in a space $\overline{\mathcal{A}}$ of generalized connections. To define $\overline{\mathcal{A}}$, consider a piecewise analytic graph $\gamma$ embedded in $\Sigma$ with edge set $E$ and vertice set $V$. The set of all such graphs is partially ordered (and directed) in the sense that for any two graphs $\gamma,\gamma'$, $\gamma\geqslant\gamma'$ if every $e'\in E(\gamma')$ can be obtained from a sequence of $e\in E(\gamma)$ by composition and reversal. The space of smooth connections associated with $\gamma$ is $\mathcal{A}_\gamma=\text{Hom}(\bar{\gamma}, G)$ where $\bar{\gamma}$ is the subgroupoid of the groupoid $\mathcal{P}$ of edges which are  equivalence classes of semianalytic curves, and $G=$SU(2) is the gauge group. For two graphs $\gamma,\gamma'$ with $\gamma\geqslant\gamma'$, one can define a projection $p_{\gamma,\gamma'}:\mathcal{A}_{\gamma}\rightarrow\mathcal{A}_{\gamma'}$. Since the set of embedded graphs are directed, one can take the projective limit of the directed family $(\mathcal{A}_\gamma,p_{\gamma,\gamma'})$ so that $\overline{\mathcal{A}}=\text{Hom}(\mathcal{P}, \text{SU(2)})$. Thus, the kinematical Hilbert space of LQG is 
\[\mathcal{H}_{\text{kin}}=L^2(\overline{\mathcal{A}},d\mu_0)\]
where $\mu_0$ is the unique Ashtekar-Lewandowski measure on $\overline{\mathcal{A}}$. Incorporating the SU(2) gauge invariance, one can further restrict the kinematical states to those gauge-covariant ones. To this end, consider the graph Hilbert spaces $\mathcal{H}_\gamma=L^2(A_\gamma,d\mu_\gamma)$ and the projections between them as isometric embeddings, $p^*_{\gamma\gamma'}:\mathcal{H}_{\gamma'}\rightarrow\mathcal{H}_\gamma$ with $\gamma\geqslant\gamma'$. Then on $\cup_\gamma\mathcal{H}_\gamma$ one can define an equivalence relation for functions $f_\gamma$ as
\[f_{\gamma_1}\sim f_{\gamma_2}\quad\text{iff}\quad\forall\gamma_3\geqslant\gamma_1,\gamma_2,~p^*_{\gamma_3\gamma_1}f_{\gamma_1}=p^*_{\gamma_3\gamma_2}f_{\gamma_2}.\]
Then the gauge covariant kinematical state space is
\[\mathcal{H}_0=\overline{\cup_\gamma\mathcal{H}_\gamma/\sim}\]
where the completion is taken with respect to the inner product induced by the directed/largest Hilbert space. Intuitively, the states in the graph Hilbert spaces can be expanded in terms of spin network states. In fact,
\[\mathcal{H}_\gamma=L^2(\text{SU(2)}^E/\text{SU(2)}^V,d\mu_{\text{Haar}})\]
where the measure is the Haar measure on SU(2). Then to each edge $e$ one assigns a spin $j^e\in\mathbb{N}/2$, the representation space of which is $V_{j^e}$ with $\dim V_{j^e}=2j^e+1$. And to each vertice $v$ is assigned an SU(2) intertwiner
\[\mathcal{I}_v:\bigotimes_{e,\text{ingoing}}V_{j^e}\rightarrow \bigotimes_{e,\text{outgoing}}V_{j^e},\quad\text{or simply}\quad\bigotimes_{v\in e}V_{j^e}\rightarrow\mathbb{C}.\]
These assignments defines a spin network state $\ket{\gamma,\{j^e\},\{\mathcal{I}_v\}}$ as an orthonormal basis for $\mathcal{H}_\gamma$, since by Peter-Weyl theorem one can expand it in terms of Wigner matrices which are orthogonal and by suitably choosing the interwtiners they can be orthonormal.

Solving  the Hamiltonian constraint in four-diemnsion gives the dynamics of LQG. Combinatorically, the dynamics of embedded spin networks in three-dimension looks like the Pachner moves of triangulations of a three-manifold. Now considering the time evolution of those embedded spin networks with Pachner-like moves, one arrives at the intuitive picture of spin foams as the path integral representation for the LQG dynamics \cite{Perez13}. From the name ``spin foam" we konw that they are (spacetime) foams colored with spins. By a foam we mean an oriented 2-cell complex consisting of 0-, 1-, and 2-cells, which can be interpreted respectively as vertices $v$, edges $e$ and faces $f$ in a complex $\Delta^*$ dual to the triangulation $\Delta$ of the manifold. The coloring is assigning spins to faces and intertwiners to edges. These coloring reduces to the coloring of spin networks when restricted to foliations. On each internal vertex $v$ of a spin foam $\sigma=(\Delta^*,j^f,\mathcal{I}_e)$, one can introduce a trace by contracting the intertwiners
\[\text{Tr}_v\sigma=\text{Tr}_v\Bigl(\bigotimes_{e,\text{ingoing}}\mathcal{I}_e^\dagger\otimes \bigotimes_{e,\text{outgoing}}\mathcal{I}_e\Bigr)\]
which intuitively gives the spin network function associated with the Pachner-like moves in the local region around $v$ \cite{KKL10}. If the foliation allows boundaries, we can similarly introduce $\partial\sigma$ and $\text{Tr}_v(\partial\sigma),v\in\partial\sigma$. Then the amplitude of the a spin foam $\sigma$ is
\[A(\sigma)=\int_{g\in\text{Irrep}(G)} dg\text{Tr}\sigma=\int_{g\in\text{Irrep}(G)} dg\text{Tr}_v\sigma\text{Tr}_v(\partial\sigma)\]
where $g$ is the irreducible representation of the gauge group $G$ assigned to a face of the foam $\sigma$.
The partition function of a spin foam model is then a sum over the spin network histories weighted by spin foam amplitudes
\[Z=\sum_\sigma W(\sigma) A(\sigma)\]
where $W$ is a measure factor.
The explicit spin foam models can be obtained from the  topological BF theory (plus  a Nieh-Yan topological term if one wants EPRL) by imposing simplicity constraints and then quantizing everything \cite{Perez13,KKL10}, or from the perturbative expansion in Feynman diagrams in group field theories (GFT) \cite{ORT15}.

It is instructive to phrase these structures in categorical terms. (See Ref. \cite{MacLane} or Appendix \ref{BBB} for categorical concepts.) As has been mentioned, the embedded graphs defining  spin networks form a groupoid $\mathcal{P}$ with vetrices as objects and edges as (iso)morphisms. Dually, the spin networks form a category with spins on edges as objects and intertwiners on vertices as morphisms. Analogously, one could say that  spin foams form a category with spins on faces as objects and intertwiners on edges as morphisms. But they do not conform to the picture of path integrals. Instead, we can take the spin networks on a foliation as objects and spin foams between them as morphisms to form a category $\mathcal{F}$. Combining $\mathcal{P}$ and $\mathcal{F}$ we see that there should be a 2-category containg both $\mathcal{P}$ and $\mathcal{F}$ \cite{Baez98}. In this paper, we show that $\mathcal{F}$ itself is a 2-category, thereby pointing to  higher categorical structures in LQG.

~

{\it Acknowledgement.} I thank Juven C. Wang for helpful comments. Financial support from National Natural Science Foundation of China under Grant Nos. 11725524, 61471356 and 11674089 is gratefully acknowledged.
\section{Review of Transition Probability Spaces}
We first recall some basic definitions and known results concerning transition probability spaces. The details can be found, for instance, in Refs. \cite{Miel68,Bel76,Can75,Can82,Gud81}.  In Sec. \ref{sec24} we show that the transition probabilities just introduced are consistent with operational probabilistic  framework for QM.

\subsection{Transition probability spaces}\label{sec21}
Let $S$ be a nonempty set of pure quantum states. In the language of Hilbert spaces, a pure state corresponds to an one-dimensional projector $x$ with a one-dimensional range spanned by unit vectors $\psi$ such that $x\varphi=\psi(\psi,\varphi)$. The transition probability between $x$ and $y$ is $p(x,y)=\text{tr}(xy)$. For any states, not necessarily pure, described by density matrices $\rho,\sigma$, the transition probability between them is the fidelity \cite{Uhl76}
\[P(\rho,\sigma)=\Bigl[\text{tr}\sqrt{\rho^{1/2}\sigma\rho^{1/2}}\Bigr]^2.\]
By observing the properties of these transition probabilities, we can define the transition probability space.
\begin{definition}\label{def21}
A transition probability space $(S,p)$ is a set of states $S$ with a real-valued transition probability $p(a,b),a,b\in S$ such that
\begin{enumerate}
\item $0\leqslant p(a,b)\leqslant 1$ and $p(a,b)=1\Leftrightarrow a=b$;
\item $p(a,b)=p(b,a)$;
\item Two states $a,b\in S$ are orthogonal if $p(a,b)=0$. A maximal subset $R\subset S$ where each distinct pair of states are orthogonal is a basis of $S$. Then \[\sum_{r\in R}p(a,r)=1,\quad\forall R\subset S,~\forall a\in S.\]
\end{enumerate}
\end{definition}
From condition (3) it is easy to see that any two bases of $S$ have the same cardinality which is just the dimension of $S$. The following basic properties resemble those of a vector space. A subspace $(S',p')$ of $(S,p)$ is such that $S'\subset S$ and $p'=p\restriction_{S'}$. Two transition probability spaces $(S_1,p_1)$ and $(S_2,p_2)$ are isomorphic if there is an isomorphism $x\rightarrow x'$ of $S_1$ to $S_2$ such that $p_1(x,y)=p_2(x',y')$. $(S_1,p_1)$ can be embedded into $(S_2,p_2)$ iff the latter contains a subspace isomorphic to $(S_1,p_1)$. The orthocomplement $A^\bot$ of a subspace $A\triangleq(S_1,p_1)$ of a transition probability space $(S_2,p_2)$ is the maximal subspace of $(S_2,p_2)$ whose states are orthogonal to $A$ and $A\cap A^\bot=\emptyset$. A subspace $A$ is orthoclosed if $A=A^{\bot\bot}$. One can easily see that an orthoclosed subspace $A$ is the orthocomplement of any basis of $A^\bot$. Moreover, the orthoclosed subspaces form an atomic orthomodular poset \cite{Bel76}. With these operations, one can indeed reconstruct the quantum logic from a transition probability space (cf. Refs. \cite{Zab75,Del84}\footnote{Note that in Ref. \cite{Del84} the symmetry condition (2) is not required. However, if one wants to represent the quantum logic on a Hilbert space, the symmetry becomes essential to prove Piron's representation theorem. Cf. Ref. \cite{Zab75} and also Ref. \cite{Bug74}.}).
Furthermore, as for vector spaces, one can construct disjoint unions (or direct sums) of $(S_1,p_1)$ and $(S_2,p_2)$ by setting $p(a,b)=0$ if $a\in S_1,b\in S_2$ or $a\in S_2,b\in S_1$ and $p(a,b)=p_i(a,b),i=1,2$ if $a,b$ belong to the same $S_i$. Then a transition probability space is irreducible if it is not the disjoint union  of two nonempty orthogonal subspaces. Every transition probability space is the disjoint union of irreducible subspaces (or superselection sectors).

In a transition probability space, one can also define a topology by introducing the  metric
\[d(a,b)=\sup_{c\in S}\bigl|p(a,c)-p(b,c)\bigr|.\]
Then $p(a,b)$ is jointly continuous in this metric topology. For arbitrary mixed states this definition of distance still holds, for which cf. Ref. \cite{MZC09}, with a change from state spaces to transition probability spaces.

An example of the transition probability space is the Hilbertian space $(S_{\mathcal{H}},p)$ where $S_{\mathcal{H}}$ is the set of all the rays $\Psi$ in a Hilbert space $\mathcal{H}$ and the transition probability between two rays $\Psi_i=e^{i\alpha_i}\psi_i,i=1,2$ is
\[p(\Psi_1,\Psi_2)=|\braket{\psi_1,\psi_2}|^2.\]
The metric can be witten in terms of the one-dimensional projectors as
\[d(\Psi_1,\Psi_2)\equiv d(x_1,x_2)=\sup_y\text{tr}[(x_1-x_2)y]=||x_1-x_2||_{\text{max}}=\sqrt{1-|\braket{\psi_1,\psi_2}|^2}\]
where in the last step we have used the relations $(x-y)^3=[1-p(x,y)](x-y)$ and $1-p(x,y)\leqslant d(x,y)$.
A subspace of $S_{\mathcal{H}}$ corresponds to a closed vector subspace of $\mathcal{H}$, which, however, does not hold in a general transition probability space. Now an important question is whether there is a representation theorem like Piron's that the transition probability space can be represented on a Hilbert space. The answer is affirmative \cite{Pul86} and the proof uses generalizations of superposition principles in quantum logics to transition probability spaces.

\subsection{Generalized transition probability spaces}
More generally, for a set $X$ with measure $\mu(X)=n$, the subsets $s\subset X$ with measure $\mu(s)=1$ form  a state space with the transition probability $p(s_1,s_2)=\mu(s_1\cap s_2)$. Following Mackey \cite{Mack63}, one can choose the  measure to be the probability measure $d\alpha_A$ on a Borel set $E$ of $\mathbb{R}$ such that $\int_Ed\alpha_A=p(A,\alpha,E)$ is the probability (amplitude) of the measurement of the observable  $A$ on the state $\alpha$ with an outcome taking value in $E$.\footnote{In the present case we use Greek letters $\alpha,\beta...$ to denote the quantum states, so as to emphasize that their measures are $d\alpha_A,d\beta_A...$.} For any pair of states $\alpha,\beta$ under a measurement of $A$, one has the transition measure (the Kakutani distance)
\[\int_Ed\sqrt{\alpha_A\beta_A}=\int_E\sqrt{\frac{d\alpha_A}{d\sigma}\frac{\beta_A}{d\sigma}}d\sigma\]
where $d\alpha/d\sigma,d\beta/d\sigma$ are the Radon-Nikodym derivatives with respect to the finite measure $\sigma$. This transition probability is generalized by Cantoni \cite{Can75}.
\begin{definition}
For any pair of states $\alpha,\beta$ under a measurement of $A$, let
\[T_A(\alpha,\beta)=\Bigl[\int_{\mathbb{R}}d\sqrt{\alpha_A,\beta_A}\Bigr]^2.\]
A generalized transition probablity from $\alpha$ to $\beta$ is
\[T(\alpha,\beta)=\inf_{A}T_A(\alpha,\beta).\]
\end{definition}
One can check that the $T(\alpha,\beta)$ satisfy the condtions (1) and (2) in definition \ref{def21}. As to the condition (3), we already have $T(\alpha,\beta)=\sup\bigl|(\alpha,\beta)^2\bigr|\equiv\mathcal{P}(\alpha,\beta)$ \cite{AR82}. Then by the concavity of $\mathcal{P}$ \cite{Uhl76}, we have
\[\sum_{\beta\in R}T(\alpha,\beta)=N\sum_{\beta\in R}\frac{1}{N}\mathcal{P}(\alpha,\beta)\leqslant N\mathcal{P}(\alpha,\frac{1}{N}\sum_{\beta\in R}\beta)=1,\]
where we have assumed, without loss of generality, that each state $\beta$ has equal probability $1/N$. The last equality comes from the fact that $R$ covers $S$ by definition and hence we have chosen $\sum_{\beta\in R} \beta=S$ by Zorn's lemma. On the other hand, since $0\leqslant T(\alpha,\beta)\leqslant 1$ and $T(\alpha,\beta)=1$ when $\alpha=\beta$, consider the special case where there is a $\beta\in R$ such that $\alpha=\beta$, then obviously
\[\sum_{\beta\in R}T(\alpha,\beta)=1+\sum_{\beta\in R,\beta\neq\alpha}T(\alpha,\beta)\geqslant 1.\]
Therefore, we must have $\sum_{\beta\in R}T(\alpha,\beta)=1$ and the $T(\alpha,\beta)$ together with the state space form a {\it generalized} transition probability space.

Furthermore, the generalized transition probability $T(\alpha,\beta)$ can be used to define a metric between pure states (the Bures distance)
\[d(\alpha,\beta)=\sqrt{2(1-\sqrt{T(\alpha,\beta)})}.\]
A space of generalized transition probabilities can be set in motion. Indeed, the action on a state $\alpha$ of an element $m$ in a mobility semigroup $\mathcal{M}$ \cite{Can82} maps the state to a fictitious state $m\alpha$ at later time. The generalized transition probability $T(\alpha,\beta)$ is nondecreasing under $m$,
\[T(\alpha,\beta)\leqslant T(m\alpha,m\beta),\quad \text{and hence}~d(\alpha,\beta)\geqslant d(m\alpha,n\beta).\]
The equality holds when $m$ represents a unitary evolution, which becomes a compatibility requirement when one wants to equip a transition probability space with a Poisson structure \cite{Lan97}. In particular, given a Poisson manifold $P$ and a space of quantum pure states $S$, one can obtain the classical limit of QM by searching for a {\it classical germ} \cite{Lan96}, i.e. a family of injections $q_\hbar:P\rightarrow S$ such that 
\[\lim_{\hbar\rightarrow0}p(q_\hbar(\chi),q_\hbar(\zeta))=\delta_{\chi,\zeta}\]
where $\chi,\zeta$ correspond to the classical (Liouville, if $P$ is symplectic) measures on $P$.

An alternative form of $T(\alpha,\beta)$ will be useful. In the quantum logic $L$, an observable $A_i$ corresponds to a proposition $a_i$ and hence the state $\alpha_{A_i}$ is mapped to the probability measure $\alpha_{a_i}$ on $[0,1]\subset\mathbb{R}$. A sequence of $a_i$ in $L$ is a maximal orthogonal sequence if they are mutually orthogonal and $\cup_ia_i=1$. Since $\mathbb{R}$ can be partitioned by Borel sets $E_i$ as $\mathbb{R}=\cup_iE_i$, we can therefore write  $T(\alpha,\beta)$ as
\[ T(\alpha,\beta)=\Bigl[\inf_{\{a_i\}}\sum_i\sqrt{\alpha_{a_i}\beta_{a_i}}\Bigr]^2.\]
In fact, one can show that \cite{Pul89}
\[ T(\alpha,\beta)=\Bigl[\inf_{R}\sum_{\gamma\in R}\sqrt{p(\alpha,\gamma)p(\beta,\gamma)}\Bigr]^2.\]
where the $p$ are the transition probabilties defined in definition \ref{def21} and the infimum is taken over all basis $R$ of $S$.

\subsection{Transition amplitude spaces}
The generalized transition probability is not the end of the story. In QM  the more computable and more intuitive quantities are the transition amplitudes, the generating functionals for path integrals. Hence, in addition to the transition probability spaces, one can also define the transition ampiltude spaces \cite{GP87}.  In later sections of this paper, the transition amplitude is the starting point of most discussions.

Consider the map $A: S\times S\rightarrow\mathbb{C}$, where $S$ is still the space of pure states, such that for two $a,b\in S,a\neq b$ to be orthogonal it only requires that $A(a,b)=0$. 
\begin{definition}\label{def23}
An $M\subset S$ is an $A$-set if the following conditions are satisfied,
\[\sum_{c\in M}|A(a,c)A(c,b)|<\infty,\quad A(a,b)=\sum_{c\in M}A(a,c)A(c,b),\quad\forall a,b\in S.\]
Let $\mathcal{N}_A$ be the class of $A$-sets.
The map $A$ given above is a transition amplitude if (i) $\mathcal{N}_A\neq0$; (ii) $A(a,a)=1,\forall a\in S$. The space $(S,A)$ is then a transition amplitude space.
\end{definition}
A transition amplitude space is strong  if $A(a,b)=1$ implies $a=b$, and total if $\mathcal{N}_A=R$, a basis of $S$. One can readily check that for a total transition amplitude space $(S,A)$, $|A(a,b)|^2$ is a transition probability and $(S,|A(a,b)|^2)$ forms a transition probability space.

As an example of transition amplitude spaces, consider again the Hilbertian space $(S_{\mathcal{H}},p)$ with the transition amplitude defined by $A(a,b)=\braket{\psi(a),\psi(b)}=\sqrt{p(a,b)}$. Amplitude functions $A_a:S\rightarrow\mathbb{C}$ such that $A_a(b)=A(a,b)$ correspond to wave functions of QM which form a complex Hilbert space. In fact, the transition amplitude spaces can also be related to the algebraic approach to QM.  Given a transition amplitude space $(S,A)$, a general map $B:S\times S\rightarrow\mathbb{C}$ is called a form. A sequence $(\mathfrak{a}_i,a_i)\in\mathbb{C}\times S$ is a null sequence $n_0(A)$ if 
\[\sum_i|\mathfrak{a}_iA(a_i,b)|<\infty,\quad\sum_i\mathfrak{a}_iA(a_i,b)=0,\quad\forall b\in S.\]
A form $B$ is an $A$-form if 
\[\sum_i\mathfrak{a}_iB(a_i,b)=\sum_i\mathfrak{a}_iB(a,b_i)=0,\quad\forall b\in S,~(\mathfrak{a}_i,a_i)\in n_0(A).\]
A form $B$ is bounded if there exists a $\mathfrak{b}>0$ such that, 
\[\sum_{c\in M}\Bigl|\sum_i\mathfrak{a}_iB(c_i,c)\Bigr|^2\leqslant \mathfrak{b}^2\sum_i|\mathfrak{a}_i|^2,\quad \forall M\in\mathcal{N}_A,~(\mathfrak{a}_i,c_i)\in\mathbb{C}\times M.\]
In the set $\mathcal{B}(S,A)$ of bounded $A$-forms, which is obviously a complex linear space under pointwise additions and scalar multiplications, one can define the involution $B^*(a,b)$ by complex conjugation, the norm $||B||$ by the infimum of the bound $\mathfrak{b}$, and the product by
\[BC(a,b)=\sum_{c\in M}B(c,b)C(a,c),\quad\forall B,C\in \mathcal{B}(S,A),~M\in\mathcal{N}_A.\]
Then  $\mathcal{B}(S,A)$ is a unital $C^*$-algebra with identity $A$. There exists an isometric $C^*$-isomorphism from  $\mathcal{B}(S,A)$  to that of bounded linear operators on the Hilbert space $\mathcal{H}$ such that $B(a,b)=\braket{\hat{B}\psi(a),\psi(b)}$, where $\mathcal{H}$ is the GNS Hilbert space with its states represented by the  positive linear functional $B_A:\mathcal{B}(S,A)\rightarrow\mathbb{R}$, or the real\footnote{The reality is due to the $A$-form condition.} amplitude functions.

\subsection{A view from operational quantum theory}\label{sec24}
We have seen that the symmetry condition of the transition probability can be relaxed if one only wants to recover the formal Born's rule of QM. However, on the one hand, the Hilbert space representation theorem requires the symmetry \cite{Zab75,Bug74}, and on the other hand, due to the timeless feature of quantum gravity the time reversal symmetry should be preserved at the level of transition probabilities.  In fact, recent results \cite{OC15,OC16} have shown that such a time reversal symmetry is permitted in a generalized formulation of operational probabilistic theories. We seek the relation between these formulations, which will be useful in the next section.

The operational formulation of QM given in Refs. \cite{OC15,OC16} consists of the following ingredients:
\begin{enumerate}
\item{\it Operations or vertices}. Each vertex is associated with an operation defined as a set of events $\{M_i^{\mathcal{I}\mathcal{O}}\}$ where the superscripts $\mathcal{I}$ and $\mathcal{O}$ represent the boundary systems within the Hilbert space $\prod_{a\in\mathcal{I}}\mathcal{H}_a\otimes\prod_{b\in\mathcal{O}}\mathcal{H}_b$,\footnote{In the time asymmetric case, $\mathcal{I}$ and $\mathcal{O}$ are respectively interpreted as {\it input} and {\it output} states.} the subscript denotes the possible values of outcomes, and the $M$ are positive semidefinite operators acting on the boundary Hilbert space.
\item{\it Wires}. Two vertices are connected by a wire representing the entangled states $W$ between the boundary systems. The composition of two operators is obtained by tracing over the entangled states between these two operations.
\item{\it Networks}. The vertices connected by wires form a network. The joint probability for this network is the probability conditioned on all the ``input" states
\[P=\frac{\text{tr}\bigl[W(\otimes M_i)\bigr]}{\text{tr}\bigl[W(\otimes\sum_i M_i)\bigr]}\]
where the indices for boundary systems are omitted for simplicity.
\end{enumerate}
Now for the formulation in terms of transition probabilty spaces, (i) the operators in operations can be chosen as the projectors $P_{\mathcal{S}}$ to the subspaces $\mathcal{S}$ of the total Hilbert space $\mathcal{H}$ so that at each vertex there is a class of states (with different phases); (ii) the wires are defined by tracing over the intermediate states $W$, which can be encoded into the summation condition of Definition \ref{def23}; (iii) the probabilty now becomes the conditional probability of finding states in the subspace $\mathcal{S}$ provided the subspace $\mathcal{S}'$,
\[P=\frac{\text{tr}\bigl[WP_{\mathcal{S}}\bigr]}{\text{tr}\bigl[WP_{\mathcal{S}'}\bigr]}.\]
If $\mathcal{S}\subset\mathcal{S}'$, one can readily check that this probability is the transition probability from states in $\mathcal{S}'$ to those in $\mathcal{S}$. Thus, the transition probability spaces are consistent with the operational formulation of QM of Refs. \cite{OC15,OC16}. This form of transition probability also shows that the expectation value of an operator in the Hilbert space language is now a transition probablity with the final state (space) being projected out by this operator.

Besides, it has been observed in Categorical Quantum Mechanics \cite{CoeckeQP} that such an operational theory equipped with compositions of operations by connecting the wires and with tensor products between operations/wires can be described by a monoidal category. For transition probability spaces, the (tensor) product $\otimes$ is definable for all the state spaces via  the disjoint union $\sqcup$ (or the coproduct $\coprod$, categorically speaking), and hence two wires can be braided making it a braided monoidal category. If two state spaces can be related by a transition, they constitute a pair of boundary systems as above. Then  the symmetry condition for transition probabilities indicates the existence of transposing of a pair of boundary systems, which is an Hermitian conjugate $\dagger$ in the tensor product Hilbert space. Hence one obtains a dagger (or rigid) monoidal category. 
\section{Transition Probability Spaces for Loop Quantum Gravity}\label{sec3}
The transition probability spaces reviewed in the last section are designed for QM only. In order to study quantum gravity one has to generalize them in two directions: the first is to include general covariance; the second is to incorporate the formalism of quantum field theories (QFT). 

We first consider the general covariant QM in both canonical and covariant formulations. Then we proceed to the general boundary formulation which is not only suitable for describing QFT but also for QM and quantum gravity. The EPRL spin foam model, defined on a given 2-complex,  can be put into a simplified general boundary formulation so that the structure of transition probability space will become manifest. Here is a conceptual conflict that for the spin foam models on 2-complexes the general covariance is not as applicable as in the continuum limit.  Here we should point out that the above two generations are {\it separate} ones and the goal is to show the structures of transition probability spaces, instead of fully characterize spin foams in terms of the general boundary formulation. These two generalizations will be undoubtfully of ultimate concern for quantum gravity, though.

In the rest two subsections, the canonical LQG Hilbert space in the discrete form is reconstructed from spin foam amplitudes and the possibility to derive quantum logics from the transition probability spaces of spin foams is pointed out.

\subsection{General covariant quantum theory I: spacetime states}
In the general covariant quantum theory proposed by Reisenberger\&Rovelli \cite{RR02}, the basic concepts are the spacetime-smeared representation of quantum states and the propagators (or transition amplitudes) between these spacetime states. Here the spacetime smearing means that both spatial and temporal measurements in realistic experiments are impossible to be infinitely sharp and need to be smeared. And by general covariance we mean that the spacetime is not fixed and can change with the spacetime smeared states through the smearing functions.

Consider a time-dependent quantum state $\ket{X,T}$ in some Hilbert space $\mathcal{H}$, where $X$ stands for spatial coordinates and $T$ for the time coordinate. Then  any such state can be transformed to another through the propagators $W$  as
\[\ket{X,T}=\int dX'dT'\braket{X',T'|X,T}\ket{X',T'}\equiv\int dX'dT'W(X',T';X,T)\ket{X',T'}.\]
In general, the following function is also a wavefunction of QM,
\[\Psi_f(X,T)\equiv\braket{X,T|\Psi_f}=\int dX'dT'W(X,T;X',T')f(X',T')\]
where $f(X,T)$ is a smooth function depending on spacetime. This is because by Riesz's representation theorem, there always exists a $\ket{\Phi}\in\mathcal{H}$ such that $f(X,T)=\braket{X,T|\Phi}$. Therefore, we can define
\begin{definition}[Reisenberger\&Rovelli \cite{RR02}]
Given a spacetime dependent smearing function $f(X,T)$ and a $\ket{X,T}\in\mathcal{H}$, the spacetime-smeared representation of $\ket{X,T}$ is
\[\ket{f}\equiv\int dXdTf(X,T)\ket{X,T}.\]
In particular, for a small (compared to the scale in question) spacetime region $\mathcal{R}$, the spacetime-smeared representation of a quantum state supported on $\mathcal{R}$ (or simply a spacetime state on $\mathcal{R}$) is
\[\ket{\mathcal{R}}\equiv C_\mathcal{R}\int_\mathcal{R} dXdT\ket{X,T},\quad  C_\mathcal{R}=\Bigl(\int_\mathcal{R}dXdT\int_\mathcal{R}dX'dT'W(X,T;X',T')\Bigr)^{-1/2}.\]
\end{definition}
The Hilbert space $\mathcal{H}_f$ of the spacetime-smeared states $\ket{f}$ has a simple inner product
\[\braket{f|f'}=\int dXdT\int dX'dT'f^*(X,T)W(X,T;X'T')f(X',T').\]
On the other hand, the spacetime states $\ket{\mathcal{R}}=C_\mathcal{R}\ket{f_\mathcal{R}}$, where $f_\mathcal{R}$ is a characteristic function of the region $\mathcal{R}$, will give two inner products. The first one is the inner product between a spacetime state  and a particle state $\braket{\mathcal{R}|\Psi}$, then $|\braket{\mathcal{R}|\Psi}|^2$ is the probability of finding the particle state $\ket{\Psi}$ in the region $\mathcal{R}$. The second is the transition amplitude between spacetime states
\[\braket{\mathcal{R}|\mathcal{R}'}=C_\mathcal{R}C_{\mathcal{R}'}\int_\mathcal{R}dx\int_{\mathcal{R}'}dyW(x;y),\quad x\equiv(X,T),~y\equiv(X',T').\]
That this is indeed a transition amplitude can be easily checked as in the following lemma.
\begin{lemma}
Given two spacetime states $\ket{\mathcal{R}},\ket{\mathcal{R}'}$ on the regions $\mathcal{R},\mathcal{R}'$ respectively, the inner product $\braket{\mathcal{R}|\mathcal{R}'}$ defines a transition amplitude $A(f_\mathcal{R},f_{\mathcal{R}'})$.
\end{lemma}
\begin{proof}
Denote $\braket{\mathcal{R}|\mathcal{R}'}\triangleq A(f_\mathcal{R},f_{\mathcal{R}'})$. Firstly, it is obvious from the definitions that $A(f_\mathcal{R},f_{\mathcal{R}})=1$. Then
\begin{align*}
&A(f_\mathcal{R},f_{\mathcal{R}'})=C_\mathcal{R}C_{\mathcal{R}'}\int_\mathcal{R}dx\int_{\mathcal{R}'}dyW(x;y)=\\
=&\int_{\mathcal{R}''}dzC_\mathcal{R}C_{\mathcal{R}'}\int_\mathcal{R}dx\int_{\mathcal{R}'}dyW(x;z)W(z;y)=\\
=&\int_{\mathcal{R}''}dzC_\mathcal{R}C_{\mathcal{R}'}\int_\mathcal{R}dx\int_{\mathcal{R}'}dyW(x;z)C_{\mathcal{R}''}C_{\mathcal{R}''}\int_{\mathcal{R}''}dz\int_{\mathcal{R}''}dzW(z;z)W(z;y)=\\
=&\int_{\mathcal{R}''}dz\Bigl(C_\mathcal{R}C_{\mathcal{R}''}\int_\mathcal{R}dx\int_{\mathcal{R}''}dzW(x;z)\Bigr)\Bigl(C_{\mathcal{R}''}C_{\mathcal{R}'}\int_{\mathcal{R}''}dz\int_{\mathcal{R}'}dyW(z;y)\Bigr)=\\
=&\int_{\mathcal{R}''}dzA(f_\mathcal{R},f_{\mathcal{R}''})A(f_{\mathcal{R}''},f_{\mathcal{R}'})
\end{align*}
where in the second line we have used the property of propagators, in the third line inserted $A(f_{\mathcal{R}''},f_{\mathcal{R}''})=1$, and in the fourth line $W(z;z)=1$ by definition. Secondly, since  $0\leqslant W(x,y)\leqslant1$ by definition, we always have $0\leqslant A(f_\mathcal{R},f_{\mathcal{R}'})\leqslant1$, and hence
\[\int_{\mathcal{R}''}dz\lvert A(f_\mathcal{R},f_{\mathcal{R}''})A(f_{\mathcal{R}''},f_{\mathcal{R}'})\rvert=\int_{\mathcal{R}''}dz A(f_\mathcal{R},f_{\mathcal{R}''})A(f_{\mathcal{R}''},f_{\mathcal{R}'})=A(f_\mathcal{R},f_{\mathcal{R}'})\leqslant1.\]
Therefore, the class  $\mathcal{N}_A$ of $A$-sets $\{\ket{\mathcal{R}''}\}$ is nonempty.
\end{proof}
Denote by $\mathcal{H}_{\mathcal{R}}$ the space of spacetime states, then $(\mathcal{H}_{\mathcal{R}},A(f_\mathcal{R},f_{\mathcal{R}'}))$ is a transition amplitude space and the following is immediate.
\begin{theorem}
 $(\mathcal{H}_{\mathcal{R}},\lvert A(f_\mathcal{R},f_{\mathcal{R}'})\rvert^2)$ is a transition probability space.
\end{theorem}
\begin{proof}
First we need to show that $(\mathcal{H}_{\mathcal{R}},A(f_\mathcal{R},f_{\mathcal{R}'}))$ is total. This can be achieved, as in the prescription for the time arrival problem given in Ref. \cite{RR02}, by choosing a sequence of spactime regions $\mathcal{R}_n=[X_n,X_n+n\epsilon]\times[T_n,T_n+n\xi]$ with finite spacetime resolution $\epsilon,\xi$ which are obviously mutually orthogonal. By Zorn's lemma we have $\mathcal{N}_A=R$. Now conditions (1) and (2) of Definition \ref{def21} are clearly satisfied.\footnote{In condition (2), since the modulus square  makes the probability spacetime symmetric, the problem of {\it retrodictions} is avoided at the level of probabilities.} Finally, since $(\mathcal{H}_{\mathcal{R}},A(f_\mathcal{R},f_{\mathcal{R}'}))$ is total, 
\[\sum_{\mathcal{R}'\in R}\lvert A(f_\mathcal{R},f_{\mathcal{R}'})\rvert^2=\sum_{\mathcal{R}'\in\mathcal{N}_A}\lvert A(f_\mathcal{R},f_{\mathcal{R}'})\rvert^2=\sum_{\mathcal{R}'\in\mathcal{N}_A} A(f_\mathcal{R},f_{\mathcal{R}'})A(f_{\mathcal{R}'},f_{\mathcal{R}})=1\]
where we have used the fact that $A(f_\mathcal{R},f_{\mathcal{R}'})=A^*(f_{\mathcal{R}'},f_{\mathcal{R}})$ which again comes from the property of $W$. Hence the condition (3) is also satisfied.
\end{proof}

For spacetime states $\ket{\mathcal{R}}$ the smearing function is just the characteristic function $\chi_{\mathcal{R}}$ of region $\mathcal{R}$. In analogy to the definition of  generalized transition probabilities, one can understand the probability $|\braket{\mathcal{R}|\Psi}|^2$ defined by the former type of inner product as follows. The inner product $\braket{\mathcal{R}|\Psi}$ can be decomposed as
\[\braket{\mathcal{R}|\Psi}=\int_\mathcal{R}dXdTC_\mathcal{R}\chi_\mathcal{R}\braket{X,T|\Psi}=\Bigl(\int_\mathcal{R}C_\mathcal{R}\Bigr)\Bigl(\chi_\mathcal{R}dXdT\Psi(X,T)\Bigr)\equiv\int_{C_\mathcal{R}}d\mathcal{R}_\Psi,\]
which represents the probability (amplitude) of finding the particle wavefunction $\Psi$ in the spacetime state localized on the small region $\mathcal{R}$ taking value in interval normalized by $C_\mathcal{R}$. Therefore, a generalized transition probability for spacetime states can be analogously defined by taking the infimum of the transition measures over the particle states $\Psi$.\footnote{In contrast to nonrelativistic QM, we have used the particle state $\Psi$ instead of an operator to study the observation. This is in fact a feature of transition probability spaces where we have shifted from operators acting on states to transition probabilities between states. Indeed, in general covariant QM, the time ordering of multiple measurements is lost and instead, it is encoded into the relations between those measurement devices (see, e.g. Ref. \cite{HMPR07}). So even if one has the operators representing those devices, one still need to know the relations between them to avoid paradoxes. Hence such a shift is suitable for general covariant QM.}

One might face the complication that in addition to the transitions between spactime states, there will be transitions between particle states localized either in the same region or in distinct regions. However, in the former case, the spacetime dependence can be discarded, hence neither the states nor the transitions are general covariant and the problem is reduced to that of the nonrelativistic QM. For the same reason, in the latter case, since the spacetime regions are distinct, transition bewteen them are prohibited if general covariance is assured.

Now there are three levels of  structures. The first is the space of smearing functions, $\mathcal{E}=\{f\}$, on which one can define a bilinear form by group averaging \cite{RR02}
\[(f',f)_C=\int d\tau\int dx[f'(x)]^*e^{i\tau C}f(x)\]
where $C$ is the constraint on a presymplectic manifold. The second is the Hilbert space $\mathcal{H}_f$ of the spacetime-smeared states $\ket{f}$, the inner product of which can be obtained from $(f,f')_C$ through the propagator $W$, 
\[(f',f)_C=\int dx'\int dx[f'(x)]^*W(x',x)f(x)=\braket{f'|f}.\]
These actually form a Gelfand triple $\mathcal{E}\subset\mathcal{H}\subset\mathcal{H}_f$ where $\mathcal{H}$ is the Hibert space of unsmeared states. The general covariant relativistic QM at this level is well-accepted and standardized in Rovelli's book \cite{Rov04}.
The third is  the space $\mathcal{H}_{\mathcal{R}}$ of spacetime states, whose inner product $\braket{\mathcal{R}|\mathcal{R}'}$ makes it a transition amplitude space. Although we have worked with this older (non-standard) version of covariant QM, they are  equivalent to each other in the categorical sense. More explicitly, take  $\mathcal{E}$ as a subcategory of the category {\bf Set} consisting of small sets (of $f$) and funcions (or inner products) between them, 
and suppose, for the moment (c.f. Sec. \ref{sec4}  for a categorification), that the state subspaces in the Hilbert space $\mathcal{H}_f$ (or $\mathcal{H}_{\mathcal{R}}$) together with the transition probabilities form a category {\bf Hf} (or {\bf Hr}). Then consider the following diagram,
\begin{diagram}
{\bf Hr}        &           &        &\\
\dTo_{i} &\rdTo^u  \rdTo(4,2)^v    &        &\\      
{\bf Hr}        &\rTo^{s}   &{\bf Hf}    & \rTo^{a}   &{\bf Set} 
\end{diagram}
where $i$ is the identity functor, $s: {\bf Hr} \rightarrow{\bf Hf} $ is the functor that effectively forgets the dependence on the time coordinate, and $a:{\bf Hf} \rightarrow{\bf Set} $ is a representable functor. Note that by the Yoneda lemma the right Kan extension $v\equiv\text{Ran}_i(a\circ s)$ of the largest triangle diagram is valid, hence  this diagram is commutative and defines a point-wise right Kan extension $u\equiv\text{Ran}_i( s)$. This $u$ is the limit of $(\{\ket{\mathcal{R}}\}\downarrow{\bf Hr})\rightarrow{\bf Hr} \rightarrow{\bf Hf}$, and since $i$ is just the identity functor and hence fully faithful, the Kan extension is universal and defines a natural isomorphism between pair of functors $\epsilon:u\circ i=u\cong s$, which simply means that
the categories $ {\bf Hr}$ and ${\bf Hf} $ are categorically equivalent.
\subsection{General covariant quantum theory II: timeless path integrals}
The general covariant QM considered above (or in Ref. \cite{Rov04}) is in a canonical formulation. When passing to the path integral formultion, the timeless feature of relativisitc QM obscures the construction. In Ref. \cite{Chiou13} this difficulty is overcame by group averaging over a time parameter with which the transition amplitude is written as a path integral similar to the case with time. For simplicity, we start directly from this timeless path integral.

Let $\mathcal{C}$ be the configuration space with coordinates $\{q^a\}$. The cotangent space $T^*\mathcal{C}$ of $\mathcal{C}$ is the phase space with coordinates $\{q^a,p_a\}$. On a curve ${\gamma}$  in $T^*\mathcal{C}$ one can define an action $S[\gamma]=\int_\gamma p_adq^a$. A physical motion is obtained by extremizing $S[\gamma]$ over those $\gamma$ on the constraint surface  defined by $H=0$.   In quantum theory the coordinates become partial obervables $\hat{q}^a,\hat{p}_a$ with respective eigenstates $\ket{q^a},\ket{p_a}$. The constraint operator $\hat{H}$ defines a projector $\hat{P}=\int d\tau\exp[-i\tau\hat{H}]$ from a kinematical Hilbert space to the physical Hilbert space.
\begin{definition}[Chiou \cite{Chiou13}]
The transition amplitude $W(q^a,q^{\prime a})=\braket{q^a|\hat{P}|q^{\prime a}}$ from $\ket{q^{\prime a}}$ to $\ket{q^a}$ is the following path integral
\[W(q^a,q^{\prime a})=N\int\mathcal{D}q^a\int\mathcal{D}p_a\delta[H]\exp\Bigl[\frac{i}{\hbar}\int_{\gamma}p_adq^a\Bigr]\]
where $N$ is a normalization factor, $\delta[H]$ confines the curve to the constraint surface such that it connects $q^a$ and $q^{\prime a}$ when restricted to $\mathcal{C}$, and the functional measures are defined as in the ordinary case, (assuming $N$ partitions of $[0,\tau]$)
\[\int\mathcal{D}q^a\triangleq\prod_{n=1}^{N-1}\int d^dq^a_n,\quad \int\mathcal{D}p_a\triangleq\prod_{n=1}^{N-1}\int \frac{d^dp_{na}}{(2\pi\hbar)^d},\quad d=\text{dim}\mathcal{C}.\]
\end{definition}
Now the properties of path integrals manifests the following.
\begin{lemma}
The path integral $W(q^a,q^{\prime a})$ is a transition amplitude.
\end{lemma}
\begin{proof}
We give a formal demonstration here.  First, $W(q^a,q^a)=1$ is manifested from the origianl expression, $\braket{q^a|\hat{P}|q^{ a}}=\braket{q^a|q^{a}}=1$ if $\ket{q^a}$ is a normalized physical state. Second, consider the composition of two paths $\gamma_2\cup\gamma_1$ where $\gamma_1$ connects $q^a$ to $q^{\prime a}$ and $\gamma_2$ connects $q^{\prime a}$ to $q^{\prime\prime a}$ when restricted to $\mathcal{C}$. Then by definition, the sum over all possible middle points $q^{\prime a}$ is equivalent to calculating the path integral from $q^a$ to $q^{\prime\prime a}$, namely,
\[\int_{q^{\prime a}}W(q^a,q^{\prime a})W(q^{\prime a},q^{\prime\prime a})=W(q^a,q^{\prime\prime a})=N'\int\mathcal{D}q^a\int\mathcal{D}p_a\delta[H]\exp\Bigl[\frac{i}{\hbar}\int_{\gamma_1\cup\gamma_2}p_adq^a\Bigr].\]
Now the convergence of the path integral is an elusive issue. In most cases the normalization factor $N$ will render $0\leqslant W\leqslant 1$.
\end{proof}
A mathematically rigorous proof of this lemma will depend on the specific mathematical approach one chooses to define path integrals.  In fact, these properties of $W$ have already been utilized in the last subsection. One can define the probability of finding $\ket{q^a}$ supported on a spacetime region $\mathcal{R}$ given the $\ket{q^{\prime a}}$ on $\mathcal{R}'$ to be
\[P(q^a,\mathcal{R};q^{\prime a},\mathcal{R}')=\Bigl\lvert \frac{W(\mathcal{R},\mathcal{R}')}{\sqrt{W(\mathcal{R},\mathcal{R})}\sqrt{W(\mathcal{R}',\mathcal{R}')}}\Bigr\rvert^2,\]
where
\[W(\mathcal{R},\mathcal{R}')=\int_\mathcal{R}dq^a\int_{\mathcal{R}'}dq^{\prime a}W(q^a,q^{\prime a}).\]
Other than spacetime regions, $\mathcal{R}$ could also be the regions in the spectra of $q^a$,
then we have a  result similar to Theorem 3.3. with a similar proof,
\begin{theorem}
The physical Hilbert space ($\mathcal{H}_f$) projectd out by $\hat{P}$ together with the probability $P(q^a,\mathcal{R};q^{\prime a},\mathcal{R}')$ defines a transition probability space.
\end{theorem}
\subsection{General boundary fomulation of quantum theories} The general boundary formulation  of quantum theory proposed in Ref. \cite{oe03a} is designed to describe a quantum theory in terms of state spaces on boundaries of regions of spacetime and transition amplitudes between them. This formulation inherits the structures of a topological quantum field theory (TQFT), but they are also suitable for describing non-topological theories such as QM. There is a caveat that in quantum gravity a smooth manifold structure is not expected for quantum geometries, so that there is no obvious connection between the general boundary formulation and quantum gravity. Here we simplify the full-fledged formulation as in \cite{oe16} to  a ``skeleton" version with the strucutre of transition probability retained. We call it the skeleton general boundary fomulation (sGBF). 

In this simplified version (with many recent refinements neglected), the general boundary formulation consists of regions $M$ of spacetime and their boundary hypersurfaces $S\subset\partial M$.  On each $S$ is defined a space $\mathcal{H}_S$ of states, and on each region $M$ is assigned an amplitude functional $Z(M)[\psi]\in\mathbb{C}$ for $\psi\in\mathcal{H}_S$. Notice that the data $(\mathcal{H}_S,Z(M))$ are not smooth manifolds in sGBF, we can envison them as embedded in the rgions of a spacetime manifold such as embedded spin networks.
We require $(\mathcal{H}_S,Z(M))$ satisfy the following (TQFT) axioms (1), (2) and the quantization axiom (3):
\begin{enumerate}
\item {\it involution}. The dual space $\mathcal{H}_S^*$ of $\mathcal{H}_S$ is obtained by changing the orientation of $S$. Now the amplitude $Z(M)$ defines a map $Z(M):\mathcal{H}_S\rightarrow\mathcal{H}_S^*$, or a cobordism $Z(M)=\braket{\mathcal{H}_S^*,\mathcal{H}_S}$.
\item {\it multiplication}. For disjoint components of $S=\cup_iS_i$ where  $S\subset\partial M$ is a codimension 1 surface with no boundary, the state space is a tensor product of each, $\mathcal{H}_S=\otimes_i\mathcal{H}_{S_i}$. For two spacetime regions $M_1,M_2$ sharing a common boundary $S_3$, if $M_1$ has another boundary $S_1$ and $M_2$ a $S_2$, the gluing $M=M_1\cup_{S_3} M_2$ gives $Z(M)=Z(M_2)\circ Z(M_1)$, where $Z(M_1):\mathcal{H}_{S_1}\rightarrow\mathcal{H}_{S_3}$ and $Z(M_2):\mathcal{H}_{S_3}\rightarrow\mathcal{H}_{S_1}$. If $S_3=\emptyset$, then $Z(M)=Z(M_2)\otimes Z(M_1)$.
\item {\it holographic quantization}\footnote{In the recent works of Oeckl \cite{oe16}, the quantization axiom is not included in the axioms of general boundary formulation. Here we work with the Schr\"odinger representation obtained by Feynman's path integral as the quantization axiom/postulate, as is originally used in \cite{oe03a}, with the aim to apply in spin foam models as a proof of principle.}. 
The space $K_S$ of field configurations $\phi$ on $S$ with an action $\mathcal{S}$ determines a classical solution inside $M$. Then $\mathcal{H}_S$ is the space $C(K_S)$ of of complex valued functions on $K_S$. For any state $\psi\in\mathcal{H}_S$, the amplitude {\it functional} is
\[Z(M)[\psi]=\int_{K_S}\mathcal{D}\phi_0\psi(\phi_0)\Bigl(\int_{\phi|_S=\phi_0}\mathcal{D}\phi e^{\frac{i}{\hbar}\mathcal{S}[\phi]}\Bigr)\equiv\int\mathcal{D}\phi_0\psi(\phi_0) Z(M)\]
where $\phi_0\in K_S$ and $\phi$ is the field configurations inside $M$ and $Z(M)$ is the amplitude map that evaluates to boundary Hilbert spaces.
\end{enumerate}
Note that the functorial axiom of TQFT \cite{Atiyah} is implicit in the assignment of state spaces and amplitudes. One can readily check that the generic data $(C(K_S),Z(M))$ in (3) satisfy the requirement of (1) and (2), and hence this formulation goes beyond a TQFT. In the following, we still work with the manifold symbols $M,S$, but we should  keep in mind that they should be understood as the smooth manifolds on which the discrete data live.

For a quantum theory described by a Hilbert space, such as QM and QFT, one can construct the following general boundary formulation. The regions are chosen as $\Sigma\times[t_1,t_2]$ where $\Sigma$ is the spatial part and the time interval indicates that the boundaries are the time-slices $S_1,S_2$. Then obviously $\mathcal{H}_{S_1}=\mathcal{H}_{S_2}^*$ and for a pair of states $\psi\in\mathcal{H}_{S_1},\phi\in\mathcal{H}_{S_2}$, the amplitude is the usual transition amplitude $\braket{\phi|U(t_2,t_1)|\psi}$ where $U$ is the unitary evolution operator. Now the difference come in the holographic quantization: the configuration space $K_S$, the way of choosing $C(K_S)$ and the representation of amplitudes. In the case of QFT, the spacetime regions $M$s are those in the Minkowski spactime and the boundaries $S$s are general hypersurfaces, not necessarily spacelike (because of background independence), in the Minkowski spacetime. 

In the general boundary formulation of quantum theory, probabilities can be defined conditionally \cite{oe07}. Let $\mathcal{P}\subset\mathcal{H}_S$ be the subspace of states representing the preparations, and $\mathcal{O}\subset\mathcal{H}_S$ be the observations. Then the probability of observing $\mathcal{O}$ given the preparations $\mathcal{P}$ with  $\mathcal {O}\subset \mathcal{ P}$ is
\[P(\mathcal{O}|\mathcal{P})=\frac{|Z(M)\circ P_{\mathcal{P}}\circ P_{\mathcal{O}}|^2}{|Z(M)\circ P_{\mathcal{P}}|^2}=\frac{\sum_i\lvert Z(M)[\psi_i]\rvert^2}{\sum_j\lvert Z(M)[\phi_j]\rvert^2}\]
where $P_{\mathcal{P}},P_{\mathcal{O}}$ are projectors onto the subspaces and $\psi_i,\phi_j$ are orthonormal bases of $\mathcal{O}$ and $\mathcal{P}$ respectively. Suppose the usual case of time-slices  as above, then the conditional probability recovers the conventional form $P(\mathcal{O}|\mathcal{P})=|\braket{\phi|U(t_2,t_1)|\psi}|^2$ if we choose as cobordisms in $\mathcal{H}_{S_1}\otimes\mathcal{H}_{s_2}$, $\mathcal{P}=\phi\otimes\mathcal{H}_{S_2}$ and $\mathcal{O}=\mathcal{H}_{S_2}\otimes\psi$. In fact, this probability  defines a transition probability  between states on different boundaries.
\begin{theorem}\label{thm37}
Given a spacetime region $M$ with boundaries $S_i$, the discrete data $(\mathcal{H}_{S_i},P(.|.))$ define a transition probability space.
\end{theorem}
\begin{proof}Now the state spaces should be the class $\{\mathcal{H}_{S_i}\}$ of boundary spaces of states. Then conditions (1) and (3) of Definition \ref{def21} are easily satisfied. (Cf. Ref. \cite{oe07}.) The symmetry condition is satisfied because of the crossing symmetry in defining the amplitudes.
\end{proof}

To see that the sGBF includes both QM and quantum gravity, consider the following formal examples.
\begin{example}
An immediate example is the transition amplitude between two arbitary  states $\ket{\psi_\alpha}$ and $\ket{\psi_\beta}$ (not necessarily being those eigenstates $\ket{q^a}$)  in general covariant QM,
\begin{align*}
W(\psi_\beta,\psi_\alpha)=&\braket{\psi_\beta|\hat{P}|\psi_\alpha}=\int dq^a\int dq^{\prime a}\braket{\psi_\beta|q^a}\bra{q^a}\hat{P}\ket{q^{\prime a}}\braket{q^{\prime a}|\psi_\alpha}\equiv\\
\equiv&\int dq^a\int dq^{\prime a}\psi^*_\beta(q^a)\psi_\alpha(q^{\prime a})W(q^a,q^{\prime a})
\end{align*}
which is exactly the quantization axiom (3). From this, we clearly see the necessity of smearing a quantum state over a spacetime region (or a hypersurface here) in order that the general covariant QM can be formulated in the general boundary formulation.
\end{example}
\begin{example}
In LQG, the $n$-point function is defined in the general boundary formulation \cite{MR05}. Consider a surface $\Sigma$ in flat spacetime, bounding a compact non-flat region $\mathcal{R}$, and approximate the n-point function by replacing the action $S[g]$ outside $\mathcal{R}$ with linearized ones. Let $\gamma$ be the value of the field on $\Sigma$, $W_\Sigma[\gamma]$ the internal integral defined by
\[W_\Sigma[\gamma]=\int_{\varphi\restriction_\Sigma=\gamma}\mathcal{D}\varphi e^{S_{\mathcal{R}}[\varphi]}, \quad \text{supp}\varphi\subset\mathcal{R},\]
 and $\Psi_\Sigma[\gamma]$ the outside integral defined alike, then the n-point function is
\[W(x_1,...,x_N)\sim\int\mathcal{D}\gamma W_\Sigma[\gamma]\gamma(x_1)...\gamma(x_N)\Psi[\gamma]\equiv\braket{W_\Sigma|\gamma(x_1)...\gamma(x_N)|\Psi_\Sigma}.\] 
Here we have two regions intersecting the real boundary $\Sigma$, but the integral still satisfies the quantization axiom (3).
\end{example}

\subsection{Spin foams} A spin foam model defines the transition amplitudes, or partition functions, between 3-geometries as a  ``sum over geometries".
\[Z=\sum_{\mathfrak{m}\in\mathfrak{M}}W(\mathfrak{m})A(\mathfrak{m})\]
where the $\mathfrak{m}$ are spin foam ``molecules" constituting the spin foam (see Ref. \cite{ORT15} or Appendix \ref{AAA} for the definition), the $W$ are the measure factors of $\mathfrak{m}$ and the $A$ are the geometric amplitudes of $\mathfrak{m}$.

 Despite the geometric forms, the spin foam amplitudes $Z$ should keep the properties of transition amplitudes, or $n$-point functions, between spin network states. In fact, they are exactly the transition amplitudes in a  group field theory (GFT) and the group fields correspond to the spin networks on  3-geometries \cite{PR11,Mik01}. At the formal level, this is well illustrated by the second quantization formulation of GFT \cite{Ori16}. Therein the spin networks $(\Gamma,\{j\},\{I\})$, on which the canonical LQG is based, are cut into disjoint {\it open spin network vertices} $(V,j)$ whose edges can be glued again to recover the original spin networks. On these spin network vertices are assigned a function of group elements in $G$,\footnote{Here the valences of all vertices are taken to be the same, that is, we are considering $d$-simplical graphs. The extension to arbitrary valent complexes are possible and can be found in Ref. \cite{ORT15}.}
\[\varphi(g_1^1,g_1^2,...,g_1^d,...,g_v^1,g_v^2,...,g_v^d)\equiv\varphi(\vec{g}_1,...,\vec{g}_v),\quad d=\dim\{j\},~v=\dim\{V\}.\]
The kinematical cylindrical functions constituting the Hilbert space of canonical LQG can be obtained by group averaging  $\varphi$ over the group $G(=$SU(2) usually). 
Then the Hilbert space of LQG can be written as a direct sum of individual Hilbert spaces on open spin network vertices,
$\mathcal{H}_{\text{kin}}=\oplus_{v\in V}\mathcal{H}_v$.
Now by second quantization we mean that the function $\varphi(\vec{g}_1,...,\vec{g}_v)$ is a $v$-partite many-body wave function with each single particle state living in $\mathcal{H}_v$, so that each many-body state can be decomposed into products of single particle states,
\[\ket{\varphi}=\sum_{i\in V}\varphi(\vec{g}_1,...,\vec{g}_v)\ket{\vec{\chi}_1}\ket{\vec{\chi}_2}...\ket{\vec{\chi}_v}\]
where $\braket{\vec{g}|\vec{\chi}}$ is the wave function for a single spin network vertex. This makes $\mathcal{H}_{\text{kin}}$  into a Fock space, and the occupation number is just the number of edges attaching an open spin network vertex so that one can define the occupation number representation,
\[\braket{\vec{g}|\varphi}=\sum_{i\in V}\varphi(\vec{g}_1,...,\vec{g}_v)\prod_i\braket{\vec{g}_i|\vec{\chi}_i}\equiv\sum_{n_i}C(n_1,...,n_a,...)\braket{\vec{g}|n_1,...,n_i,...}.\]
 As in the usual second quantization scheme, one can introduce creation and annihilation oprators $a^\dagger_{\vec{\chi}},~a_{\vec{\chi}}$ generating these occupation number states, and the resulting field operators $\hat{\varphi}=\sum_{\vec{\chi}}a_{\vec{\chi}}\varphi$ are just the field operators in GFT. In this sense the  transition amplitudes will take the conventional QM form \cite{Mik01,Ori16}
\[Z=\braket{\psi|\varphi}=\braket{\varphi|e^{iS}|\varphi}=\sum_{\Gamma}W(\Gamma)A(\Gamma),\]
where in the last expression  the $\Gamma$ is the perturbative Feynman diagram in GFT and one has $A(\Gamma)=A(\mathfrak{m})$ for some $\mathfrak{m}$. Hence all the conditions for tansition probability spaces are formally satisfied.

There is a conceptional difficulty here that in GFT there is no spacetime and hence no boundaries, and so are the spin foams constructed from GFT. As we have emphasized in the last (sub)sections, the transition probability space approach focuses on the quantum states and the transition probabilities between them. In this sense, the boundaries are subset of states in the state space. Also from the sGBF one can see that it is possible in principle to embedd the discrete spin networks/foams into a smooth manifold.

 At a more sophisticated level, in order to find the structures of transition probability spaces, we can put the spin foam models into the general boundary formulation.
  Consider the slicing of a closed spin foam $\sigma$ into two open spin foams $\sigma_1,\sigma_2$ with $\partial\sigma_1=\partial\sigma_2=\Gamma$. Then on the boundary spin networks $\Gamma$ there are boundary vertices $\bar{v}$ and boundary edges $\bar{e}$. The spin foam sum with boundaries is (schematically)
\[Z_{\sigma,\partial\sigma}=\sum_{\text{internal }(f,e,k)}\prod_{f,e,k} W_{f,e,k} \prod_{v\notin\partial\sigma} A_v(j,\mathcal{I},k)\big|_\partial\]
where the $k$ label the wedge links, i.e. the vertex-face pair $(vf)$ labeling the subdivision of a face into wedges consisting of vertices, half edges and internal edges designated to define the bivectors in EPRL  \cite{CF08}. When one glues $\sigma_1,\sigma_2$ back to a closed $\sigma$, the boundary vertices and edges become respectively internal edges and faces. Therefore
\[Z_{\sigma=\sigma_1\cup_\Gamma\sigma_2}=\sum_{\partial(f,e,k)}\prod_{\partial(f,e,k)} W_{\partial(f,e,k)} Z^*_{\sigma_2,\partial\sigma_2}Z_{\sigma_1,\partial\sigma_1}\big|_\partial\equiv\braket{Z_{\Delta_2}|Z_{\Delta_1}}.\]
Such a $Z$ thus preserves the composition of cobordisms and hence they form a corbordism category. Indeed, in the EPRL spin foam model, states on the boundary 3-geometries match the projected spin network states, and the amplitude functional  can be written as a path integral
\[Z_\sigma=\sum_{\partial(f,e,k)}\prod_{\partial(f,e,k)} W_{\partial(f,e,k)}Z_{\sigma,\partial\sigma}\Psi_P\] 
where $\Psi_P$ is the projected spin networks. This obviously takes the form of the general boundary formulation. We refer to Ref. \cite{CF08} for these details (although the distinction between the continumm and discrete GBF is not made there).

Hence the EPRL spin foam model admits a general boundary formulation with the boundary state space being that of projected spin networks and with $Z(M)$ the spin foam sum matching the boundary data. Then as a corollary of Theorem \ref{thm37} in the last subsection, we see that they also defines a transition probability space. 

Alternatively, one can define  transition probability spaces by first identifying the transition amplitude spaces. In EPRL, the spin foam sum is indeed a transition amplitude. An important observation is that the stationary amplitude equals identity since the contractable loops are physically equivalent to the base point.\footnote{This is because one must fix the redundant gauge degrees of  freedom or incorporate diffeomorphism invariance into the path integral measures of spin foams so as to get an anomaly-free spin foam quantization. See Ref. \cite{BP10} for arguments based on toy models.} The preservation of cobordisims by the spin foam sum ensures the requirement of compositions of amplitudes. This is also in line with the holonomy representation of spin foams \cite{MP12} where the spin foam amplitude can be written as the product of vertex amplitudes and {face} amplitudes
 \[Z=\sum_{\mathfrak{m}}\int_Gdh_{vf}\prod_vA_v\prod_f\delta\Bigl(\prod_{v\in f}h_{vf}\Bigr),\quad v,f\in\mathfrak{m}\]
where $h_{vf}$ is the holonomy along the wedge link. Since the face amplitudes are invariant under local face cuttings (or compositions), 
\[\int dh_{ext}\delta(h,...,h_{ext})\delta(h',...,h_{ext})=\delta(h,...,h'...),\quad h_{ext}: \text{holonomy along the cut},\]
then the compositions of amplitudes are assured.
The finiteness of the amplitude sum can be achieved by choosing suitable normalization of the spin foam sum (see, e.g. Ref. \cite{MV13}). Hence the class of A-sets $\mathcal{N}_A$ is nonempty. Furthermore, the resulting transition amplitude space is total as in the proof of theorem 3.3.. A problem might arise as whether the form of transition probabiltiy is the  conditional $P(.|.)$. In fact, the combinatorial structures of spin foams \cite{ORT15} can be readily fit into the formulation of the operational quantum theory,\footnote{Such an argument needs futher detailed explanations from which we will refrain. Pointedly stated,  the vertices are open spin network vertices in the second quantization formulation of GFT; the wires are edges to be connected; the joint probability is given by the spin foam sum.} so that the transition probability should be defined conditionally (within a cobordism). 

Therefore, in either approach one obtains the following result.
\begin{theorem}
In the EPRL spin foam model, the data $(\mathcal{H}_{S,p},P(.|.))$ consisting of the spaces $\mathcal{H}_{S,p}$ of (projected) spin network states on the boundary $S$ and transition probability $P(.|.)$, define a transition probability space.
\end{theorem}

\subsection{Reconstruction of the canonical formalism} We have briefly argued that the spin foam amplitudes are indeed transition amplitudes in the sense of Definition \ref{def23}. Based on this, we can try to reconstruct the canonical formalism as in the discussion below Definition \ref{def23}. The same idea has already been exploited in Ref. \cite{PR11} where the spin foam amplitudes are interpreted as the quantum gravitational analogues of Wightman functions in ordinary QFT,  and the  spin network states as quantum states in QFT. Moreover, the GNS construction is applied to the $C^*$ algebra on the space $\mathcal{A}$ of spin network states together with linear functionals on $\mathcal{A}$. This  is a  general result for spin foam models, and also holds for the EPRL spin foam model we have discussed above. However, we will see that the  transition amplitudes provide a simpler approach, in that we no longer need the QFT analogues to apply the GNS construction.

More explicitly, to define the null sequence one first needs to introduce as before the amplitude functions $A_{s_1}(s_2)=A(s_1,s_2)\equiv Z(s_1,s_2)$ where the $s_i$ denote the spin network states. Then the null sequence condition becomes
\[\sum_i\mathfrak{a}_iA(s_1,s_{2,i})=\sum_i\mathfrak{a}_iA_{s_1}(s_{2,i})=A_{s_1}\Bigl(\sum_i\mathfrak{a}_is_{2,i}\Bigr)=0\quad\Rightarrow\quad\sum_i\mathfrak{a}_is_{2,i}=0.\]
If one uses the Dupuis-Livine map \cite{DL10} from the usual SU(2) spin network states to those projected spin network states,
\[s\mapsto \tilde{s}(g)=\int_{SU(2)}dhK(g,h)s(h)\Rightarrow\sum_h\mathfrak{k}_hs_{h},\quad g\in SL(2,\mathbb{C}),~h\in SU(2)\]
where in the sum we have supposed a discrete version of this map just to compare the notations, which is wrong since the group is continuous.
Then  the null sequence condition can be interpreted as the closure constraint imposed on the projected spin networks. Indeed, if one further takes the coherent state representation of SU(2) spin network states \cite{LS07}, then after group-averging   over SU(2), the  spin network state $\ket{j,n}$ in the so-called spin-normal representation is lablled by the normals $n$ assinged to the triangles of the tetrahedra which satisfy the closure condition. Besides, the map above is reduced to a discrete one in the coherent state representation as desired. 

Meanwhile, the spin foam amplitudes are naturally $A$-forms, since in spin foam models transitions are just relational and hence interchanging  $s_1$ and $s_2$ gives the same amplitude
\[\sum_i\mathfrak{a}_iA(s_1,s_{2,i})=\sum_i\mathfrak{a}_iA(s_{2,i},s_1)=0\]
for the null sequence $(\mathfrak{a}_i,s_{2,i})$. The bound $\mathfrak{b}$ is definable as in Ref. \cite{MV13}. Then we have the space $\mathcal{B}$ of bounded linear amplitude functions or $A$-forms. The product is defined as the composition of amplitudes, the involution as the complex conjugation, and the bound as $||A||=\inf\mathfrak{b}$ as before. These make $\mathcal{B}$ a $C^*$ algebra, and again there is a $C^*$ isomorphism from $\mathcal{B}$ to those bounded linear operators in the Hilbert space that can be chosen as
the GNS construction from  a discerete version of  holonomy-flux algebra $\mathfrak{P}$ of spin network states.

Let us turn to the structures of superselection sectors. For a transition probability space $(S,p)$, its superselection sectors are those irreducible subspaces $(S_i,p_i)$, the disjoint union of which compose the whole space $\sqcup_i(S_i,p_i)=(S,p)$. It is superselected in the sense that $p_i(S_i,S_j)=0$ for $i\neq j$. In the context of LQG, spin network states supported on distinct graphs are orthogonal, and hence kinematically one has the direct sum decomposition, $\mathcal{H}_{\text{kin}}=\oplus_\gamma\mathcal{H}_\gamma$. At the dynamical level, the superselection sectors can be defined again by the vanishing of transition probabilities or spin foam amplitudes. In principle, this can be done either in the non-graph-changing or graph-changing case, since in both cases the spin foam amplitudes are definable. This way the pertinent Hilbert space can be rendered separable.

\subsection{Logic of quantum gravity and entropy}\label{sec36}
Any quantum theory of gravity presuming the correctness of QM remains  a theory of QM, thereby the structures of quantum logic should be preserved. We have seen in the above subsections that LQG, especially the spin foam models, can be described alternatively by transition probability (or amplitude) spaces. Then constructing the quantum logic of LQG is now a conceptually trivial task, which is mathematically the same as  those in Ref. \cite{Del84,Pul90}.\footnote{
Note that obtaining quantum logics from transition probabilities seems to be different from the procedure of deriving quantum logics from the multidimensional probabilities or $S$-probabilities \cite{BM91} for the state space $S$. However, in analogy to the fact that the transition measures \cite{Can75,Mack63} defining the generalized transition probablities are constructed from the probability measures for two single events,  these reconstructions of logical structures are indeed equivalent.
} What we want to do is to  find the quantum logics directly from the structures of LQG. To this end, one could start from the Hilbert space of cylindrical functions in  canonical LQG. Mathematically, this should not be different from those in the Hilbert space representation theorem for quantum logics \cite{Piron64,Bug74,Pul86}. We will, inspite of being less rigorous, consider physically the spin foams as 4D geons in Appendix \ref{AAA}.

At the semantic level, such a reconstruction gives us the (realistic) information-value of a sentence in the logic from the probability-value of that sentence. However, this is not an emergence as it seems to be. In view of the construction of quantum logical probability \cite{HSP14}, quantum probability and quantum logic can be both defined on a lattice, so that our reconstruction only gives a complementary view of the underlying mathematical/physical structures. In particular, we show in the following  these two views are not  the same at scales.

Denote by $L_p$ the quantum logic of LQG obtained from the transition probability spaces. To define ``states" in  $L_p$ (for the definition of which see Appendix \ref{BBB}), consider a basis $R$ of the state space $S$ and its mutually orthogonal subdivision $R=\sqcup_iR_i$. Define the coarse-grained transition probability as $\sum_{r\in R_i}p(a,r)\equiv p(a,R_i),\forall a\in S$, then the map $\mu_p:R\rightarrow p(a,R_i)$ is a ``state" in $L_p$. In the specific $L_p$ of LQG, the quantum logical ``states" correspond to coarse-grained bulk spin foams. 
Although there are already many schemes of coarse-graining in spin foam models, the coarse-grained transition probability obtained from the general boundary transition amplitudes forgets much of the details of the interior fine-grained dynamics. So in order to distinguish non-equivalent coarse-grained spinfoams, let us turn to entropy functionals.

Recall that on a quantum logic $L$, entropy can be defined for a partition of $L$ \cite{Yuan05}. By a partition  of $L$, we mean a set $Q=(a_1...,a_i,...a_n)$ of elements $a_i\in L$ such that $\mu(\lor_{i=1}^na_i)=1$ and $Q$ is join-orthogonal, i.e. $(\lor_{i=1}^ka_i)\bot a_{k+1}$ for $1\leqslant k\leqslant n-1$.  A ``state" has the Bayesian property if $\mu(b)=\sum_i\mu(a_i)\mu(b|a_i),\forall Q,\forall a_i\in Q,b\in L$ where $\mu(b|a)$ is the conditional quantum logical ``state". For two partitions $Q,Q'$ of $L$, the join or {\it common refinement} $Q\cup Q'$ is still a partition of $L$ if the ``state" has the Bayesian property. Now one can define, for example, the simplest (Shannon-type) entropy of the partition $Q$ with respect to $\mu$,
\[H_\mu(Q)=-\sum_{i=1}^n\mu(a_i)\log\mu(a_i).\]
The subadditivity $H_\mu(Q\cup Q')\leqslant H_\mu(Q)+H_\mu(Q')$ holds if $\mu$ has the Bayesian property.

In $L_p$ a partition $Q_p$ can be defined as the subdivision of the basis  $R=\sqcup_iR_i$. Then by the definition of transition probabilities (or amplitudes), $Q_p$ has the Bayesian property,\footnote{
In the modern treatment of unsharp quantum measurements in terms of effect algebras, the Bayesian property could be violated, which requires additional structure such as sequential products. In the current case, the quantum logic $L_p$ inherits the properties of transition probabilities so that the Bayesian property is satisfied. Physically speaking, the causality  fixes the decoherence histories in the evolutions of spin foams, which in turn is presented by definable transition probabilties. Cf. Sec.\ref{sec43}.
} and hence the subadditivity of the entropy functionals for common refinement of a family of partitions $\{Q_{p1},...,Q_{pi},...,Q_{pk}\}$ of $L_p$. To proceed, let the partitions $Q_{pi}$ of $L$ be generated by a ``state"-preserving coarse-graing map $\mathfrak{C}:L_p\rightarrow\{Q_{pi};i=1,2,...,k\};Q_{p}\mapsto\mathfrak{C}_{(i)}(Q_{p})=Q_{pi}$, and $\mu_p(\mathfrak{C}(R_i))=\mu_p(R_i)$. Then for a partition $Q_p$, it is clear that $H_{\mu_p}(\mathfrak{C}(Q_p))=H_{\mu_p}(Q_p)$. Next, define the Kolmogorov-Sinai entropy
\[h(\mathfrak{C})=\sup_{\{Q_{pi}\}}h(\mathfrak{C},Q_{p}),\quad\text{where}~ h(\mathfrak{C},Q_{p})=\lim_{n\rightarrow\infty}\frac{1}{n}H_{\mu_p}(\bigcup_{i=1}^kQ_{pi}).
\]
The limit above exists due to the subadditivity of $H_{\mu_p}$. 

In Sec.\ref{sec21} we have defined isomorphic transition probability spaces in the fine-grained sense. Now let $\mathfrak{I}:S_1\rightarrow S_2$ be such an isomorphism of two state spaces $S_1$ and $S_2$ that preserves the transition probabilities. Taking partitions into consideration, we further require that $\mathfrak{C}_2(\mathfrak{I}R_i)=\mathfrak{I}\mathfrak{C}_1(R_i),\forall R_i\subset R\subset S_1$.
\begin{proposition}
Let $(S_i,p_i),i=1,2$  be two fine-grained transition probability spaces isomorphic to each other. If on each space is given a coarse-graining generator $\mathfrak{C}_i$ of partition on the quantum logic $L_{pi}$ reconstrcted from $(S_i,p_i)$, then $h(\mathfrak{C}_1)=h(\mathfrak{C}_2)$.
\end{proposition}
\begin{proof} If there is only one partition $Q_{p1}$ of $L_{p1}$, then by acting the isomorphism $\mathfrak{I}$ one sees that $\mathfrak{I}Q_{p1}$ is the only partition of $L_{p2}$ and $H_{\mu_p}(Q_{p1})=H_{\mu_p}(\mathfrak{I}Q_{p1})$.

In the case of multiple partitions $Q_{p1},Q_{p2},...,Q_{pk}$ generated by $\mathfrak{C}(Q_p)$, one has the subadditivity 
\[H_{\mu_p}(Q_{p1}\cup Q_{p2}\cup...\cup Q_{pk})\leqslant H_{\mu_p}(Q_{p1})+H_{\mu_p}(Q_{p2})+...+H_{\mu_p}(Q_{pk}).\]
Since the ways of joining these partitions are generally different, one can choose a common refinement such that  $H_{\mu_p}(\bigcup_{i=1}^k Q_{pi})$ takes the smallest value $H_{min}(Q_p)$. Then by the subadditivity, 
\[H_{min}(Q_p)+\epsilon>H_{\mu_p}(\bigcup_{i=1}^k Q_{pi}),\quad\forall\epsilon>0.\]
Since $\mathfrak{I}$ is probability-preseving, one has $H_{\mu_p}(\bigcup_{i=1}^k Q_{pi})\geqslant H_{min}(\mathfrak{I} Q_{pi})$, which means there exists $h(\mathfrak{C}_1,Q_{p})$ defined by the smallest $H_{\mu_p}(\bigcup_{i=1}^k Q_{pi})\equiv H_{min}(Q_p)$ such that
\[h(\mathfrak{C}_1,Q_{p})\geqslant  h(\mathfrak{C}_2,\mathfrak{I}Q_{p})\quad\Rightarrow\quad h(\mathfrak{C}_1)\geqslant  h(\mathfrak{C}_2,\mathfrak{I}Q_{p}).
\]
For single $Q_p$ we have shown that $\mathfrak{I}Q_{p1}=Q_{p2}$, hence $h(\mathfrak{C}_1)\geqslant h(\mathfrak{C}_2)$. The reverse direction is similar.
\end{proof}
We thus see that if $h(\mathfrak{C}_1)\neq h(\mathfrak{C}_2)$, the underlying two fine-grained transition probability spaces cannot be isomorphic to each other. This on one hand improves the usefulness of the quantum logical ``states", and on the other hand shows the advantage of  transition probability spaces. Furthermore, entropy functionals establish equivalence classes of coarse-grained transition probabilities, which when stated in the language of spin foams removes unphysical gauge freedoms from the spinfoam path integral. In addition to temporal or dynamical relations (as logical connectives) in spin foams there are also ``spatial" gauge relations (as value-equivalent propositions), thereby the gauged spinfoam path integrals form an algebra of the Lindenbaum-Tarski type. This motivates us to dive into 2-categories in the next section.

\section{Categorical Characterizations}\label{sec4}
In this section we explore possible categorical characterizations of  transition probability spaces and spin foams. 

Before we proceed, we should point out that in the following categorical discussions, the probabilistic nature of the original transition probability space is diminished, while the transition amplitudes as logical connectives play the major role. Then  is it still meaningful in playing with categories, if the probabilistic attribute of quantum mechanics is absent? In fact, the categorical abstraction leads us to the more fundamental structure of {\it Constructor Theory} \cite{D13}. A category  consists of morphisms and likewise the constructor theoy is about tasks which are sets of transformations.  Without the probabilities, the statement about the probability of the possible or the impossible becomes those simply about the possible or the impossible, which is the basic principle of Constructor theory. The question of why they are possible  or  impossible is answered by the (quantum) spacetime itself as a constructor. The probabilities are attributes instantiated in the subsidiary theory of concern.

Although the cases to be considered are not new,  {
 they will be interesting if we keep the fundamental constructor structure in mind. It turns out that transition probability spaces realize many proposals by Crane  for quantum gravity.

\subsection{Categorification of a transition probability space}
The categorical characterization of transition probabilities was first considered by Uhlmann in Ref. \cite{Uhl85}. Therein the objects and morphisms (or arrows) are defined respectively as the states and the transition probabilities between states. For two different sets $A$ and $B$ of states that are causally related, or simply for two $C^*$ algebras connected by a completely positive map $K$, the $t$-functor is defined by exactly the  map $K$. Let us denote this category by {\bf TPS}. 

In {\bf TPS} the hom-set $hom(a,b)$ of  morphisms  between two objects $a$ and $b$ has at least two arrows, since given a usual transition probability one can always construct Cantoni's generalized transition probability from it. 
In this sense, one can define the category {\bf TP}$(a,b)$ of transition probabilities between states $a$ and $b$ with the different transition probabilities $T$ (in the generalized sense, with  measures replaced by valuations if necessary) as objects and the transformations between them as arrows. {\bf TP} is not necessarily the two-object category {\bf2}, since there is a large set of distance measures for quantum states \cite{MZC09}. Then {\bf TP} also satisfy the conditions in Definition \ref{def21}, since every object in {\bf TP} satisfies them, which means that {\bf TPS} is a 2-category with states as 0-cells, transition probabilities as 1-cells, and arrows in {\bf TP} as 2-cells, i.e.
\[
\begin{diagram}[w=3em]
~\\
\mathcal{H}_a & \upperarrow{p}
\lift{-2}{\ \ \Downarrow{t}}
\lowerarrow{T} & \mathcal{H}_b& \upperarrow{p'}
\lift{-2}{\ \ \Downarrow{t'}}
\lowerarrow{T'} & \mathcal{H}_c,\\
~
\end{diagram}\quad t,t'\in{\bf TP}.
\]

On the other hand, the spin foam also defines a category {\bf SF}\footnote{We use this notation to differ it from the $\mathcal{F}$ in the Introduction since it is now a 2-category. } where the objects are spin network states on quantum 3-geometries and the morphisms are the {\it equivalence classes} of spinfoams relating the spin network states supporting different representations. 
In Ref. \cite{Baez98}, the equivalence relation is defined by (i) affine transformations; (ii) subdivisions; (iii) orientation reversal of the spin foams. A moment of reflection shows that the operation of subdivisions is just the condition (3) of Definition \ref{def21}, and the orientation reversal is the condition (2). The affine transformations do not change the cells of the complexes on which the spin foams are based, so that the transition probabilities or the spin foam amplitudes remain intact. In order to identify the 2-cells in {\bf SF}, consider the description of spin foams via branched coverings  as in Ref. \cite{DMY10}. Recall that a branched covering of a three-manifold $W$ is a submersion $p:W\rightarrow S^3$ of $W$ into the three-sphere such that the covering is trivial when the graph $\Gamma$ embedded in $S^3$, the branched locus, is deleted. Two branched covering of three-manifolds $W_0,W_1$ along two distinct graphs  $\Gamma_0,\Gamma_1$ can be related by a branched cover cobordism $M$ such that $\partial M=W_0\otimes W_1^*$. The branched covering of $M$, $q:M\rightarrow S^3\times[0,1]$, reduces to those of $W_i,~i=0,1$ when restricted to $S^3\times\{0\}$ or $S^3\times\{1\}$. In the context of LQG, the spin networks $(\Gamma,j,\mathcal{I})$ embedded in  a three-geometry $W_i$ play the role of brached loci. For two spin networks based on the graphs $\Gamma,\Gamma'$, there is a branched covering of some three-manifold $W$ connecting them,
\[\Gamma\subset S^3\leftarrow W\rightarrow S^3\supset\Gamma'\]
which is in fact the 2-cell relating the two spin foams emanating from $\Gamma,\Gamma'$. The concept of cobordism we have used above to give the transition amplitude $Z(M)$ is in fact the  four-manifold $M$ such that there is a branched cover cobordism defining the spin foam,
\[\Sigma\subset S^3\times[0,1]\leftarrow M\rightarrow S^3\times[0,1]\supset\Sigma'\]
where $\Sigma$ and $\Sigma'$ satisfying respectively $\partial\Sigma=\Gamma_0\otimes\Gamma_1^*$ and $\partial\Sigma'=\Gamma'_0\otimes\Gamma_1^{\prime*}$ are branched loci of $W$. Therefore, {\bf SF} is a 2-category with the 2-cells defined by the branched covers.\footnote{ In the Introduction, we have seen that another way to get a 2-category is to take into account the category {\bf SN} of spin networks with vertices as objects and edges as arrows. Now the 2-morphisms are just spin foams connecting the spin network states. In this approach  the vertices are taken as objects, which spoils analogy with the transition probability spaces. Also notice that in Ref. \cite{DMY10} the 2-category for (top)spin foams are defined by as objects the (top)spin networks, as 1-morphism the branched coverings and as 2-morphisms(/cells) the branched covering cobordisms. Here we have exchange the definitions of 1- and 2- morphisms to fit into the framework of transition probabilty spaces. }

By the theorems of the last section, one has a faithful functor $\mathfrak{Pr}:{\bf SF}\rightarrow{\bf TPS}$ such that the description of spin foams by transition probabilities defines an equivalence of categories. However, there are properties that are not shown in transition probability spaces. For instance, the causality is not shown in transition probability spaces since they are time-symmetric: if one takes the point of view that the state space $S$ should be a $C^*$-algebra, the involution (i.e. the dagger/transposing operation) is allowed by the symmetric condition in Definition \ref{def23} (which inherits from the time-reversal symmetry in the standard quantum formalism). Causality  becomes explicit only when there are further conditional inputs as shown in Sec.\ref{sec24}.
Besides, since  the (spin network) states are evolving under transitions, the transition probabilities can not tell us whether the properties of the states are preserved. To understand these deeper structures without external inputs, let us discuss more categorical characterizations.

\subsection{Property transitions, classical germs and quantum topoi}
We have seen in \S2.4. that in order to be consistent with the operational quantum theory, one has to define the transition probability conditionally.  Let $\mathcal{H}$ be the Hilbert space of states, $\mathcal{L}(\mathcal{H})$ be the Hilbert lattice with as partial order the subspace projections in $\mathcal{H}$, and  $\mathcal{P}(\mathcal{H})$ be the power set of  $\mathcal{H}$. Then for an initial state space $\mathcal{H}$, each subspace $\mathcal{S}\in\mathcal{P}(\mathcal{H})$ represents a possible transition probability $P=\text{tr}(WP_{\mathcal{S}})/\text{tr}(WP_{\mathcal{H}})$ conditioned on $\mathcal{H}$. The following characterization of property transitions is due to Coecke and Stubbe \cite{CS99}. We describe it here for clarification of the above-mentioned problems. In this subsection we neglect the possible intuitionistic logics and still work with quantum logics.

Define the {\it operational resolution} $\mathcal{C}:\mathcal{P}(\mathcal{H})\rightarrow\mathcal{L}(\mathcal{H})$ that preserves the partial order of $\mathcal{P}(\mathcal{H})$ defined by subset inclusions. In other words, $\mathcal{C}$ is a left adjoint of the embedding functor from $\mathcal{L}(\mathcal{H})$ to the Boolean algebra formed by subset inclusions in $\mathcal{P}(\mathcal{H})$  such  that the disjunction is definable for $\mathcal{L}(\mathcal{H})$ in the ambient $\mathcal{P}(\mathcal{H})$. The reason that we have used the word ``functor" will be clear in the following.

In the current case, the operational resolution is simply the map $\mathcal{C}:\mathcal{T}\mapsto\cup\mathcal{T}$ from $\mathcal{T}$ to the join of all $\mathcal{T}\in\mathcal{P}(\mathcal{H})$. Now given two state spaces $\mathcal{H}_1$ and $\mathcal{H}_2$, one can write a general state transition as
\[f:\mathcal{P}(\mathcal{H}_1)\rightarrow\mathcal{P}(\mathcal{H}_2),\quad\text{such that}~f(\cup_i\mathcal{T}_i)=\cup_if(\mathcal{T}_i),~\forall \mathcal{T}\in\mathcal{P}(\mathcal{H}_1).\]
Notice that $f$ is probabilisitic in the sense that it maps an initial subspace in $\mathcal{P}(\mathcal{H}_1)$ into all possible subspaces in $\mathcal{P}(\mathcal{H}_2)$, and hence the join-preserving condition corresponds to $\sigma$-additivity of probabilistic measures. These maps $f$ form a join-complete lattice with as bottom element the map mapping to the empty set, and as top element the identity. 
In the physical sense, the union or disjunction of subspaces means the lacking of resolution, i.e. one is not able to distinguish the state transitions to the smaller subspace.
Hence one obtains a {\it quantaloid}\footnote{A quantaloid is a category of which every hom-set  is a join-complete lattice with the composition of joins being distributive. More precisely, it is a category enriched over a symmetric monoidal closed join-complete lattice.} $\mathcal{Q}$ with as objects the data $(\mathcal{H},\mathcal{L}(\mathcal{H}),\mathcal{C})$, and as hom-sets the complete lattice formed by $f$. Then by the propositions proved in Ref. \cite{CS99}, one has the following commutative diagram
\[
\begin{CD}
\mathcal{P}(\mathcal{H}_1) @>{f}>> \mathcal{P}(\mathcal{H}_2)\\
@VV{\mathcal{C}_1}V @VV\mathcal{C}_2V\\
\mathcal{L}(\mathcal{H}_1) @>f_{\text{pr}}>> \mathcal{L}(\mathcal{H}_2)
\end{CD}
\]
This means that the probabilisitc state transition $f$ will induce the definite property transition $f_{\text{pr}}$\footnote{Note that one cannot tell whether this is  the transition probability, hence it is only a simple map.} between the property lattices.

Notice that the state transition function $f$ is probabilistically defined on the classical Boolean algebra in $\mathcal{P}(\mathcal{H})$, whereas the transition probability is defined on the Hilbert spaces $\mathcal{H}$ of states. This can be understood by referring to the relations between the classical and quantum event structures \cite{Zaf01}. In fact, since the Boolean algebra in $\mathcal{P}(\mathcal{H})$ is defined set-theoretically,  on $\mathcal{P}(\mathcal{H})$ one can define a category of presheaves, or a topos, ${\bf Set}^{\mathcal{P}(\mathcal{H})^{\text {op}}}$ (or $[\mathcal{P}(\mathcal{H})^{\text{op}},{\bf Set}]$). Then the operational resolution $\mathcal{C}$ is the left adjoint of the embedding {\it functor} from the category of quantum event structure $\mathcal{L}(\mathcal{H})$ to the topos ${\bf Set}^{\mathcal{P}(\mathcal{H})^{\text {op}}}$, where
an object $L\in\mathcal{L}(\mathcal{H})$ is mapped to a presheave on $\mathcal{P}(\mathcal{H})$. Now the state transition $f$ are transition maps between the Boolean charts which locally cover $\mathcal{L}$. So as on manifolds, the quantum transition probabilities can be obtained from those $f$ between charts. 

Since we have the quantaloid $\mathcal{Q}$, the enriched categorical structure over $\mathcal{Q}$ \cite{Stu05} can now be exploited. First,  $\mathcal{P}(\mathcal{H})$ can be considered as $\mathcal{Q}$-sets labelled by the operational resolution $\mathcal{C}$, so that  $\mathcal{P}(\mathcal{H})$ becomes a $\mathcal{Q}$-category with  as arrows the $f$ between its  elements.  Then  a presheaf on the $\mathcal{Q}$-category $\mathcal{P}(\mathcal{H})$ is the distributor from $*_\mathfrak{h}$ to $\mathcal{P}(\mathcal{H})$ where $*_\mathfrak{h}$ is the one-object $\mathcal{Q}$-category of $\mathfrak{h}\in\mathcal{P}(\mathcal{H})$. Intuitively a distributor is a  matrix of transition functions $f(\mathfrak{a}_1\rightarrow\mathfrak{b}_2),\mathfrak{a}_1\in\mathcal{P}(\mathcal{H}_1),\mathfrak{b}_2\in\mathcal{P}(\mathcal{H}_2)$ that plays the role of an arrow  between two $\mathcal{Q}$-categories. Hence as expected the presheaves relate the local structure to the global one. Now these allow us to explain the classical germs mentioned before. 
\begin{proposition}
Given a Hilbert space $\mathcal{H}$ of states, the operational resolution $\mathcal{C}:\mathcal{P}(\mathcal{H})\rightarrow\mathcal{L}(\mathcal{H})$ induces a classical germ.
\end{proposition}
\begin{proof}
Consider again the operational resolution $\mathcal{C}$. If the quantum logic $\mathcal{L}(\mathcal{H})$ tends to a classical Boolean lattice $\mathcal{L}(\mathcal{B})$, then $\mathcal{C}$ is a set-theoretic bijection and hence one has the  bijection $a$ between the two $\mathcal{Q}$-presheaves such that
\begin{diagram}
*_\mathfrak{h}&\rTo^{a} &*\\
\dTo~{\text{distr.}} & &\dTo~{\text{ident. distr.}\in Y[\mathcal{L}(\mathcal{B})]}\\
\mathcal{P}(\mathcal{H})&\rTo^{\text{bijection}~\mathcal{C}} & \mathcal{L}(\mathcal{B})
\end{diagram}
commutes, where $Y[\mathcal{L}(\mathcal{B})]$ is the Yoneda embedding of the $\mathcal{Q}'$-sets $\mathcal{L}(\mathcal{B})$ induced by the identity distributor with the quantaloid $\mathcal{Q}'$  defined as before only with $\mathcal{L}(\mathcal{H})$ replaced by $\mathcal{L}(\mathcal{B})$. Therefore, in the classical limit the transition functions for quantum states can be identified to be the state transitions $f$ between the Boolean presheaves on $\mathcal{P}(\mathcal{H})$. If one takes the bottom $f$ as zero, then the matrix of transition function $f(\mathfrak{a}\rightarrow\mathfrak{b})=\delta_{\mathfrak{ab}}$ for $\mathfrak{a,b}\in\mathcal{P}(\mathcal{B})$ in this limit, i.e. a classical germ.
\end{proof}

Returning to spin foams, we have seen in the last section that the subspaces of concern are those subspaces of a tensor product space (or a cobordism) $\mathcal{H}_1\otimes\mathcal{H}_2$. However, since the work of D. Aerts \cite{Aerts82}, we know that separated compound quantum states can not be described by the property lattice of a tensor product of Hilbert spaces. A possible way out is to still work with the quantaloid $\mathcal{Q}$, and then to reconstruct the tensor product Hilbert space \cite{Coe00}. More explicitly, for two Hilbert spaces $\mathcal{H}_1$ and $\mathcal{H}_2$ constituting the product space  $\mathcal{H}_1\otimes\mathcal{H}_2$, one can form the quantaloid $\mathcal{Q}$ with two objects corresponding to $\mathcal{H}_1$ and $\mathcal{H}_2$. The hom-set $\mathcal{Q}(\mathcal{P}(\mathcal{H}_1),\mathcal{P}(\mathcal{H}_2))$ of state transitions  induces the property transitions $\mathcal{Q}(\mathcal{L}(\mathcal{H}_1),\mathcal{L}(\mathcal{H}_2))$ which as a complete lattice has the atomic bottom map that maps  atoms to atoms or zero. Such an atomic map further induces a linear operator $F\in B(\mathcal{H}_1,\mathcal{H}_2)$ mapping $\mathcal{H}_1$ to $\mathcal{H}_2$. Then  by the following isomorphisms,
\[\mathcal{H}'_1\otimes\mathcal{H}_2\cong B(\mathcal{H}_1,\mathcal{H}_2),\quad B'(\mathcal{H}_1,\mathcal{H}_2)\cong \mathcal{H}_1\otimes\mathcal{H}_2\]
where the prime denotes the dual space,
one recovers the product Hilbert space. Now in the case of cobordisms, the state transition between the individual entities (or subsystems) is manifested by the transition probability given in the last section, which also avoids the case of interaction-free subsystems. Therefore, in the current case the property transitions are reflected in the transitions between individual subsystems, whereas the compound state transitions can be defined directly on the compound system.

The structure we are playing with is in fact a {\it quantum topos} \cite{Cra07}, that is, a category of sheaves over a quantaloid. Over each object of the quantaloid $\mathcal{Q}$ there is the power set $\mathcal{P}(\mathcal{H})$, between the objects of which the maps $f$ define a complete lattice. In $\mathcal{P}(\mathcal{H})$ the set inclusions define a partial order and hence $\mathcal{P}(\mathcal{H})$ becomes a {\it site} (of a locale). A presheaf is a contravariant functor that takes the site to a set.
Since the state transitions $f$'s induce transition probabilities between Boolean charts, one can see that without summing over the intermediate states the $f$'s are lax:
\[f_{1,2}\otimes f_{2,3}\leqslant f_{1,3}\]
where the tensor product is in accordance with tensor product Hilbert space and the partial order is defined in the power set. This gives a category of presheaves over the quantaloid $\mathcal{Q}$. The transition prbabilities further define the gluing property for all presheaves, which results in a category of sheaves over quantaloid.

Note that although  {\bf SF} has been taken as a 2-category, in this subsection we have worked {\it inside} this 2-category. If one considers such structures on the entire 2-category, the resulting structure will be {\it (quantum) cosmos as the 2-category of stacks}, which should work well in the categorical generalization of spinfoam models, e.g. the spin-cube models based on the 2-categorical representations of the Poncar\'e 2-group \cite{Min13}.

\subsection{Causality: causal categories and causal sites}
Suppose we have a symmetric monoidal 2-category {\bf SF}, then we can include causality to form causal categories \cite{CL13} without further operational inputs. Recall that a causal category is a symmetric monoidal category whose unit object $1$ is terminal. A causal category is normalized if the terminal arrows $T_A:A\rightarrow1$ form a comma category $(\{A\},1)$ of objects over the terminal object $1$. The category {\bf SF} of spin foams  is evidently normalized if the transition probabilities are defined conditionally via Bayesian inversion. Indeed, from the conditional definition of transition probabilities, one can consider the conditional subspace inclusions as set inclusions so that the forgetful functor ${\bf SF}\rightarrow{\bf Set}$ creates colimits. Now consider the identity functor $1_{\bf SF}:{\bf SF}\rightarrow{\bf SF}$, then the existence of colimits implies that there is a left Kan extension $\text{Lan}_G1_{\bf SF}$ along the functor $G:{\bf SF}\rightarrow{\bf1}$ which is from {\bf SF} to the category of unit object,
\begin{diagram}
{\bf SF} &\lTo^{1_{\bf SF}} &{\bf SF} & \\
&\luTo_{\text{Lan}_G1_{\bf SF}} &\dTo_G \\
& &{\bf1}
\end{diagram}
 then by the formal criteria for the existence of adjoints, $\text{Lan}_G1_{\bf SF}$ is the right adjoint of $G$. From the above diagram we can infer that the colimit of $1_{\bf SF}$ is the unit object $1$, so if $G$ is taken to be the diagonal functor, then the terminal object is the unit object, $\text{Lan}_G1_{\bf SF}(1)=\text{Colim}1_{\bf SF}$=1. (To put it simply,
\[\text{Lan}_G1_{\bf SF}(1)=\int^{\sigma}{\bf1}(G\sigma,1)\cdot\sigma=\int^{\sigma}{\bf1}(\coprod_\sigma 1,1)=\underrightarrow{\lim}1_{\bf SF}.)\]

Next,  there are two ways presented in Ref. \cite{CL13} to inculde causality: one is by restrcting to the subcategory with causal structure; the other is to combine causal structure with the category. The existing ways to build  causality into spin foam models in the literature can be roughly classified into these two classes:
\begin{enumerate}
\item {\it By restricting}: one can directly put causal restrictions on the vertices of spin networks in such a way that they obey a partial order. See, e.g. Ref. \cite{Gup00}. Such a restrction is recently shown to be due to the fact that the  orientations of the {\it volume 4-vectors} should be restricted by the closure condition for the 4-simplices. So those with opposite orientation with repect to the whole manifold should be discarded. See, e.g. Ref. \cite{Imm16}.
\item {\it By combining}: as in field theory one can put  restrictions on the domain of integration over the proper time in the formal  path integral of spin foams. In Ref. \cite{Ori05}, a proper time variable is introduced into the GFT, so that restriction can be described by combining the original GFT with it.
\end{enumerate}
In either way, we obtain a causal category where the information flow in this case is represented by the definable transition probability.

There is a caveat that causal category might only be applied in categorical QM, since the creations/annihilations of particles in QFT are not define in category {\bf Hilb} of Hilbert spaces. In the current case, spin foam models are constrained TQFTs, so  loosely speaking the concept of causal category is still appropriate.

Then why are there two ways of implementing causality? On first thought, these two ways are equivalent or complementary. But it is possible for them to co-exist. Observe that the restriction of state spaces is only  necessary for implementing causality. Another partial order induced soley by causality is required. The two partial orders on the category define a bisimplicial set with causality, which has been termed as {\it causal sites} \cite{CC05}. Inside {\bf SF}, since we have defined sites for the power set $\mathcal{P}(\mathcal{H})$ for a state space $\mathcal{H}$, the complete lattice of the arrows $f$  defines the partial order of restrictions,  and at the same time they induce transition probabilities that are the causal partial order. Hence we obtain the structure of causal {\it2}-sites.


\section{Conclusion and Outlook}
We have studied the transition probability spaces in LQG and especially in spin foam models. Although they might be disadvantageous in calculating concrete quantities when compared to the orthodox quantum formalism, we have seen the advantages in studying structural issues such as the relations between canonical and covariant formalisms, and quantum logical structures. 

We have also discussed some  categorical structures of the transition probability spaces. However, these discussions are only {elementary} and can be further investigated in the following directions: (i) topos quantum theory or quantum topos; (ii) higher categories: since one can take the spin networks as 1-category so that {\bf SF} is a 3-category; the stacks over (causal) sites are possible tools to study unification theories.

On the other hand, the transition probability in this paper reminds us of the transition probability in stochastic processes. The transition probability spaces in LQG might open the window to intriguing questions like (i) rigorously constructing path integral for spin foams from stochastic processes as for conventional path integrals, and hence (ii) quantization of gravity from a stochastic point of view. (It does not deed to coincide with the established stochastic quantization technique.)

Moreover, what we have not yet studied in this work are  the various metrics or distances on transition probability spaces in LQG. These distances are useful as the measures for quantum correlations in the Hilbert space formulation. {
However, now the transition probabilities are defined by spin foam amplitudes, so these distances should measure the {\it evolving} correlations, or correlations  at a coarse-grained level. The entropy functionals on the quantum logic of LQG discussed in Sec.\ref{sec36} is an example of this.
The situation is similar to the tensor networks where the one gets one dimension higher through the renormalization group flow. But one should be careful if one wants to find  the Ryu-Takayanagi formula in this dynamical  context. On the one hand, in contrast with the kinematical case \cite{HH17}, at the level of spin foam amplitudes the structures of spin networks could be wiped out and the exact holographic mapping may not exist; on the other hand, these distances are between the transition probabilities instead of states, and the spin foam sums do not say much about the ``bulk" states or geometries. If one only considers spin networks states on boundary 3-geometries, then the causal evolution would be a fine-graining with finer triagulations instead of coarse-graing if one takes causality into account.

There is also a possible connection with braiding operations. Note that the  tensor product $\otimes$ in Sec.4.2 is used as in TQFT to combine the lower dimensional boundary manifolds of a cobordism. Although the LQG dynamics is localized on spin networks, we can consider only the spin foams when calculating transition probabilities, and hence one can further define the Gray tensor product  in the 2-category {\bf SF} of spin foams to make it monoidal. As an example the disjoint product $\sqcup$ of state spaces supported on disjoint patches of the spatial  surface, $\mathcal{H}_1\sqcup\mathcal{H}_2$, then it is possible for two spin foams first to intertwine with each other and then split into the transposed product $\mathcal{H}_2\sqcup\mathcal{H}_1$. This is a kind of braiding of two spin foams, which leads to anyon statistics even if we are in $(3+1)$-dimension. One way of describing such braidings is the loop braid group (see, e.g. Ref. \cite{BWC07}). There are of course more complicated braidings in $(3+1)$-dimensional topological phases such as three-string braidings (see, e.g. Refs. \cite{WW15}), which has a nice description via spacetime surgeries \cite{WWS16}. In the current framework, since the 2-morphisms relating the spin foams in {\bf SF} are defined by  coverings of the three-geometries branched over the spin networks, the braiding operation $\tau$ on spin foams can be described either by the braiding of the ribbon surfaces in the branching set, or by the Kirby diagrams as in the surgery theory of manifolds. In Ref. \cite{BP06} it is shown that both descriptions correspond bijectively to a braided monoidal category, so it is possible that the braiding operations are spacetime surgeries even when we are considering the quantum spacetime.

Finally, since the transition probability spaces can be interpreted operationally, we have to mention the recent beautiful work of L. Hardy \cite{Har16} where an operational framework for gravity is systematically developed. It would be interesting to study transition probability spaces and categorical structures in that framework. Otherwise, {\it ``...it is a foolish cave explorer who throws away a light"} \cite{CC05}.

\appendix

\section{Some Mathematical Definitions}\label{BBB}
Here we present some basic defintions in category theory \cite{MacLane} and quantum logics \cite{Red98} used in the main text. We go through these concepts in a brief (and informal) manner and for details one is referred to those cited works.

A category ${C}$ is a directed graph with identities and composition. The vertices $a,b,c,...$ are called objects, and the edges $f,g,h,...$ between vertices are called morphisms. If one takes a morphism of a category as an object of a new category and defines new 2-morphisms between the original 1-morphisms, one gets a 2-category. The higher categories can be defined in this way but with more complicated structures. So by a category we  mean a 2-category with objects, morphisms between objects, functors between morphisms, and natural transformations with two ways of composition (i.e. vertical and horizontal). Sometimes the vertical composition endow the hom-set a categorical strucutre $E$, and we say the category $C$ is enriched in $E$. 

A groupoid is a category with all its morphisms invertible. The invertibility of a morphism $f:a\rightarrow b$ in ${C}$ means that there is a morphism $f':b\rightarrow a$ such that $ff'=1_a,f'f=1_b$. If there is a natural transformation $\tau:C\xrightarrow{.} B$ such that the arrow $\tau c$ is  invertible in $B$, then $\tau$ is a natural isomorphism. Two categories $C$ and $B$ are equivalent if there is a pair of functors $S:C\rightarrow B,T:D\rightarrow C$ with their composites being natural isomorphisms, $1_C\cong T\circ S, 1_B\cong S\circ T$. 

A monoidal category $B\equiv(B,\otimes,e,\alpha,\lambda,\rho)$ consists of a category $B$, a bifunctor or tensor product $\otimes:B\times B\rightarrow B$, a unit object $I\in B$ and three natural isomorphisms (called
associator and left/right unitor)
\[\alpha:(b_1\otimes b_2)\otimes b_3\cong b_1\otimes(b_2\otimes b_3),\quad\lambda:I\otimes b\cong b,\quad\rho:b\otimes I\cong b\]
which satisfy many coherence conditions. Coherence can be formally understood as every diagram formed by $\alpha,\lambda,\rho$ commutes. 
The braiding in $B$ is a natural isomorphism $\tau:b_1\otimes b_2\cong b_2\otimes b_1$, and $B$ is symmetric if $\tau^2=1$. A  symmetric monoidal category of special interest in the main text is the cobordism category ${\bf Bord}_n$. A cobordism is a (smooth) $n$-dimensional manifold $M$ with disjoint boundary components $\partial M=W_1\sqcup W_2$ where the orientation of $W_2$ is opposite to that of $W_2$. The objects in ${\bf Bord}_n$ are $(n-1)$-dimensional maniflods $W$ on the boundary, the morphisms are $n$-dimensional manifolds $M$,  and the composition is the gluing of two $W$'s along the common boundary, and and the tensor product is the disjoint union $\sqcup$.  A topological quantum field theory (TQFT) is therefore a strong monoidal functor from ${\bf Bord}_n$ to a category of vector spaces such as Hilbert spaces
\[Z:{\bf Bord}_n\rightarrow{\bf Vect}.\]
Here a monoidal functor between monoidal categories $M,M'$ consists of the following tuple of functors satisfying the coherence condition with respect to the tensor product: an ordinary functor $F:M\rightarrow M'$, a bifunctor $F_2:F(a)\otimes F(b)\rightarrow F(a\otimes b)$, and a functor units $F_0:e'\rightarrow e$. A monoidal functor is strong if $F_2$ and $F_0$ are isomorphisms. If   $F_2$ and $F_0$ are not necessarily isomorphisms and satisfy no other conditions, we say the monoidal functor is lax.

Many important categorical structures requires very long descriptions and proofs. To be ultimately brief, let us jump directly to (co)ends. Consider two categories $C,B$ and two bifunctors $S,T:C^\text{op}\times C\rightarrow B$. A dinatural transformation $\alpha:S\xrightarrow{..}T$  assigns  to each $c\in C$ a morphism $\alpha_c:S(c,c)\rightarrow T(c,c)$ in $B$ such that for every $f:c\rightarrow c'$ in $C$ the following diagram
\begin{diagram}
S(c',c) &\rTo^{S(f,1)} &S(c,c)&\rTo^{\alpha_c}T(c,c)\\
\dTo^{S(1,f)} & &&\dTo_{T(1,f)} \\
S(c',c')&\rTo^{\alpha_{c'}} &T(c',c')&\rTo^{T(f,1)}T(c,c')
\end{diagram}
commutes. A wedge from $S$ to $b\in B$ is obtained by collapsing the $T$-part into an object $b$. An end of a functor $S:C^\text{op}\times C\rightarrow B$ consists of a universal wedge $\omega:e\xrightarrow{..}S$, where the universality means there exists a unique $h:x\rightarrow e$ such that the following diagram
\begin{diagram}
x&&&&\\
&\rdDashto~{h}\rdTo(4,2)^{\beta_b}\rdTo(2,4)_{\beta_c}&&&\\
& &e&\rTo_{\omega_b}&S(b,b)\\
& &\dTo_{\omega_{c}} & &\dTo_{S(1,f)}\\
& &S(c,c) &\rTo^{S(f,1)}&S(b,c)
\end{diagram}
commutes, and the constant object is usually called the end of $S$, denoted as
\[e=\int_cS(c,c).\]
Coends are dually defined as
\[e=\int^cS(c,c).\]
Then for every functor $T:C\rightarrow B$, one can form a bifunctor $S:C^{\text{op}}\times C\xrightarrow{Q}C\xrightarrow{T}B$ with $Q$ being the projection onto the second factor in such a way that the limit of $T$ is the end of $S$:
\[\underleftarrow{\lim} T\cong\int_cS(c,c).\]
The tensor product mentioned above of two functors $S:C^{\text{op}}\rightarrow B$ and $T:C\rightarrow B$ is therefore a coend
\[S\otimes_PT=\int^PS\otimes T.\]
Given a functor $K:M\rightarrow C$ and a category $A$, define the left and right Kan extensions as respectively the left and right adjoints to the functor of the functor categories $A^K:A^C\rightarrow A^M$. T hen for a functor $T:M\rightarrow A$ left and right Kan extensions of $T$ along $K$ can be  respectively expressed in terms of (co)ends as
\[\text{Lan}_KTc=\int^mC(Km,c)\cdot Tm,\quad \text{Ran}_KTc=\int_mTm^{C(c,Km)},\quad\forall m\in M,c\in C\]
where $C(c_1,c_2)\cdot c_3=C(\coprod_{c_3}c_1,c_2)$ is the copower (i.e. coproduct  with all the factors $c_1$ equal) with $C(c_1,c_2)$ being the morphisms in $C$  and $c_1^{C(c_1,c_2)}=\prod_{f\in C(c_1,c_2)}c_1$ is the power. The Kan extensions subsume all concepts in category theory; cf. Ref. \cite{MacLane}. 
For instance, the right Kan extension $\text{Ran}_{K}1_M$ of the identity functor $1_M$ along the functor $K$ is a left adjoint of $K$.
Likewise, the Yoneda lemma is the right Kan extension of $T:M\rightarrow {\bf Set}$ along the identity functor $K=1_C$,
\[Tc=\text{Ran}_{1_C}Tc=\int_mTm^{C(c,m)}=\int_m{\bf Set}(C(c,m),Tm)\cong\text{Nat}(C(c,-),T).\]
We say $T$ is representable and $Y:c\rightarrow C(c,-)$ is a Yoneda embedding. Then a right Kan extension is point-wise if it is preserved by all representable functors $A(a,-):A\rightarrow{\bf Set}$, i.e.
\[A(a,-)\int_mTm^{C(c,Km)}=\int_mA(a,-)\circ Tm^{C(c,Km)}.\]

Now let us recall some quantum logics. A quantum logic is an orthomodular lattice $L\equiv(L,\land,\lor,~^\neg,\leqslant,0,1)$ where $0$ and $1$ are respectively smallest and largest element. The orthomodularity means that $a\leqslant b\Rightarrow b=a\lor(a^\neg\land b)$ for $a,b\in L$.  An element $x\in L$ is an atom if $0<x$ and there exists no element $y\in L$ such that $0<y<x$. Then L is atomic if for every nonzero element $x\in L$ there exists an atom $a$ of $L$ such that $a\leqslant x$.
 A state on $L$ is a  map $\mu:L\rightarrow[0,1]$ such that $\mu(1)=1$ and for any family $\{a_i\}$ of pairwise orthogonal elements in $L$, $\mu(\lor_ia_i)=\sum_i\mu(a_i)$. 

We can consider other lattice structures \cite{Cra07}. A sup complete lattice is a lattice with infinite superema, and a locale is a distributive sup complete lattice. Here by complete we mean as uasual every element has a supremum. Consider a category whose elements are those of a locale and a morphism between $a$ and $b$ iff $a\leqslant b$. This defines a site of the locale. Then a presheaf on a locale is a contavariant functor $[\text{site}^{\text{op}},{\bf Set}]$, and by requiring gluing property for different sections in a presheaf one gets a sheaf. A category of presheaves is an example of a topos (which is sufficient this paper). If the site is endowed with a Grothendieck topology, then the category of sheaves over such sites is a Grothendieck topos.

\section{Spin Foams and 4-geons}\label{AAA}
In this appendix, we show that spin foams can be regarded as (qauntum) 4-geons from which quantum logics can be derived. Here the transition probability space no longer plays any role.

Geons are topologically nontrivial regions in spacetime that model the elementary particles. Usually the geons are 3D, while in Ref. \cite{Had97} a speculation is proposed that there exists 4D 4-geons such that the quantum logic can be constructed from the spacetime structures of 4-geons. Although alleged as classical, the 4-geons should be the {\it atoms} of quantum spacetime if the quantum logic is credited as the characteristic feature of a quantum theory. In the following, we argue that spin foams satisfy the  axioms defining the 4-geons.

Spin foam models can be understood as the description of dynamical evolutions of the boundary spin network states, or directly as the combination of their basic building blocks.
Therefore let us recall the combinatorial structures of spin foams \cite{ORT15}. To define a spin foam model, one starts with a set $\mathfrak{A}$ of boundary graphs, the bisected  versions $\mathfrak{B}$  of which are obtained by introducing new vertices on each edge of graphs in $\mathfrak{A}$ such that the new vertices are bivalent, i.e. a subdivions of the complexes. Connecting the bisected boundary graphs to a bulk vertex forms a spin foam {\it atom} $\mathfrak{a}$ with its faces defined by three distinct vertices. These spin foam atoms can be bound to form a spin foam {\it molecule} $\mathfrak{m}$, if the bounding faces of the atoms have the same numbers of vertices and edges. The spin foam molecules constitute the whole spin foam. With these in mind, we have the following
\begin{proposition}
The spin foam atoms are 4-geons in a quantum spacetime.
\end{proposition}
\begin{proof}
We have to show that spin foam atoms satisfy the defining axioms of 4-geons conjectured in Ref. \cite{Had97}:
\begin{itemize}
\item{\it localization}: the 4-geon spacetime should be localized in small (quantum region) and asymptotically flat;
\item{\it particle-like}: measurements in 4-geons should yield a true or a false value only, just like measuring classical particles;
\item{\it state preparation and measurement}: suitable boundary conditions for 4-geons.
\end{itemize}
 First, the micro-locality \cite{MS07} of the combinatorical structures of spin foams satisfies the localization axiom. The problem of disordered locality will not arise since we stay in the quantum regime. Second, since a generic spin foam sum depends on the vertex amplitudes, a hypothetical experiment testing the presence of a bulk vertex $v$ in a spin foam atom will give a true/false value only, or a delta function $\delta(v)$. This fulfills particle-like axiom.\footnote{In GFT, these spin foam atoms contribute to the second quantized many-body wave functions as single particle wave functions. In this sense, it is indeed particle-like.} Third, the axiom of state preparation and measurement  is now reflected in the bonding of the faces of spin foam atoms, because the bonding maps only identify  faces with the same number of  vertices and edges, which restricts the possible combinations of atoms forming a spin foam molecule. Since the spin foams themselves are results of the quantization of gravity, such a restriction effectively sets the boundary conditions for possible gravitational solutions.
\end{proof}

With this result, one can proceed as in Ref. \cite{Had97} to construct an atomic orthomodular lattice of 4-geons, which in turn is the quantum logic $L_p$ of LQG.

\bibliographystyle{amsalpha}

\end{document}